\documentclass[10pt,notitlepage,tightenlines,nofootinbib,superscriptaddress,aps,pra,twocolumn]{revtex4-2}
\pdfoutput=1

\usepackage{newpxtext,newpxmath}

\usepackage[latin1]{inputenc}
\usepackage{amsthm}
\usepackage{amssymb}
\usepackage{amsmath}
\usepackage{bbold}
\usepackage{bbm}
\usepackage[dvipsnames,svgnames]{xcolor}
\usepackage[pdftex,backref=page]{hyperref}
\hypersetup{
    bookmarksnumbered=true, 
    breaklinks=true,
    unicode=false, 
    pdfstartview={FitH}, 
    pdfnewwindow=true, 
    colorlinks=true, 
    allcolors=MidnightBlue!70!black!70!TealBlue,
}
\usepackage{braket}
\usepackage{dsfont}
\usepackage{mathdots}
\usepackage{mathtools}
\usepackage{enumerate}
\usepackage[shortlabels]{enumitem}
\usepackage{csquotes}
\usepackage{stmaryrd}
\usepackage[cal=boondox]{mathalfa}
\usepackage{graphicx}
\usepackage{stackengine}
\usepackage{scalerel}
\usepackage{tensor}       
\usepackage{array}
\usepackage{makecell}
\newcolumntype{x}[1]{>{\centering\arraybackslash}p{#1}}
\usepackage{tikz}
\usepackage{pgfplots}
\usetikzlibrary{shapes.geometric, shapes.misc, positioning, arrows, arrows.meta, decorations.pathreplacing, decorations.pathmorphing, patterns, angles, quotes, calc}
\usepackage{booktabs}
\usepackage{xfrac}
\usepackage{siunitx}
\usepackage{centernot}
\usepackage{comment}
\usepackage{chngcntr}
\usepackage[justification=justified,format=plain]{caption}
\usepackage{subcaption}
\usepackage{ragged2e}

\newtheorem{thm}{Theorem}
\newtheorem*{thm*}{Theorem}
\newtheorem{prop}[thm]{Proposition}
\newtheorem*{prop*}{Proposition}
\newtheorem{lemma}[thm]{Lemma}
\newtheorem*{lemma*}{Lemma}
\newtheorem{cor}[thm]{Corollary}
\newtheorem*{cor*}{Corollary}

\newtheorem*{cj*}{Conjecture}

\newtheorem*{Def*}{Definition}

\newtheorem*{question*}{Question}

\newtheorem*{problem*}{Problem}

\makeatletter
\def\thmhead@plain#1#2#3{%
  \thmname{#1}\thmnumber{\@ifnotempty{#1}{ }\@upn{#2}}%
  \thmnote{ {\the\thm@notefont#3}}}
\let\thmhead\thmhead@plain
\makeatother

\theoremstyle{definition}

\makeatletter
\newcommand{\manualifempty}[3]{%
  \edef\@tempa{#1}%
  \ifx\@tempa\@empty
    #2
  \else
    #3
  \fi
}
\makeatother

\makeatletter
\newtheoremstyle{manualstyle}
  {3pt}{3pt}{\itshape}{}{\bfseries}{.}{ }{}

\theoremstyle{manualstyle}

\newtheorem{manualthminner}{Theorem}
\newenvironment{manualthm}[1]{%
  \def\@tempa{#1}%
  \ifx\@tempa\@empty
    \renewcommand{\themanualthminner}{\unskip}%
  \else
    \renewcommand{\themanualthminner}{#1}%
  \fi
  \begin{manualthminner}
}{\end{manualthminner}}

\newtheorem{manualpropinner}{Proposition}

\newtheorem{manuallemmainner}{Lemma}

\newtheorem{manualcorinner}{Corollary}

\makeatother

\newcommand{\bb}{\begin{equation}\begin{aligned}\hspace{0pt}}
\newcommand{\bbb}{\begin{equation*}\begin{aligned}}
\newcommand{\ee}{\end{aligned}\end{equation}}
\newcommand{\eee}{\end{aligned}\end{equation*}}
\newcommand\floor[1]{\lfloor#1\rfloor}

\newcommand{\eqt}[1]{\stackrel{\mathclap{\scriptsize \mbox{#1}}}{=}}
\newcommand{\leqt}[1]{\stackrel{\mathclap{\scriptsize \mbox{#1}}}{\leq}}

\newcommand{\geqt}[1]{\stackrel{\mathclap{\scriptsize \mbox{#1}}}{\geq}}
\newcommand{\ketbra}[1]{\ket{#1}\!\!\bra{#1}}

\newcommand{\sumno}{\sum\nolimits}

\newcommand{\e}{\varepsilon}
\renewcommand{\epsilon}{\varepsilon}

\newcommand{\dd}{\mathrm{d}}

\newcommand{\id}{\mathds{1}}
\newcommand{\R}{\mathds{R}}
\newcommand{\N}{\mathds{N}}

\newcommand{\C}{\mathds{C}}
\newcommand{\E}{\mathds{E}}

\newcommand{\ve}{\varepsilon}

\newcommand{\sep}{\mathrm{SEP}}

\let\SEP\sep

\DeclareMathOperator{\Tr}{Tr}

\DeclareMathOperator{\co}{conv}

\DeclareMathAlphabet{\pazocal}{OMS}{zplm}{m}{n}

\newcommand{\iid}{\mathrm{iid}}

\newcommand{\HH}{\pazocal{H}}

\newcommand{\EE}{\pazocal{E}}

\newcommand{\D}{\pazocal{D}}

\newcommand{\XX}{\pazocal{X}}

\newcommand{\FF}{\pazocal{F}}

\newcommand{\lsmatrix}{\left(\begin{smallmatrix}}
\newcommand{\rsmatrix}{\end{smallmatrix}\right)}

\newcommand{\deff}[1]{\textbf{\emph{#1}}}
\newcommand{\rel}[3]{#1\big(#2\,\big\|\,#3\big)}
\newcommand{\Rel}[3]{#1\Big(#2\,\Big\|\,#3\Big)}

\stackMath
\newcommand\xxrightarrow[2][]{\mathrel{%
  \setbox2=\hbox{\stackon{\scriptstyle#1}{\scriptstyle#2}}%
  \stackunder[5pt]{%
    \xrightarrow{\makebox[\dimexpr\wd2\relax]{$\scriptstyle#2$}}%
  }{%
   \scriptstyle#1\,%
  }%
}}

\newcommand{\tends}[2]{\xxrightarrow[\! #2 \!]{\mathrm{#1}}}
\newcommand{\tendsn}[1]{\xxrightarrow[\! n\rightarrow \infty\!]{\mathrm{#1}}}

\stackMath

\makeatletter
\newcommand*\rel@kern[1]{\kern#1\dimexpr\macc@kerna}
\newcommand*\widebar[1]{%
  \begingroup
  \def\mathaccent##1##2{%
    \rel@kern{0.8}%
    \overline{\rel@kern{-0.8}\macc@nucleus\rel@kern{0.2}}%
    \rel@kern{-0.2}%
  }%
  \macc@depth\@ne
  \let\math@bgroup\@empty \let\math@egroup\macc@set@skewchar
  \mathsurround\z@ \frozen@everymath{\mathgroup\macc@group\relax}%
  \macc@set@skewchar\relax
  \let\mathaccentV\macc@nested@a
  \macc@nested@a\relax111{#1}%
  \endgroup
}

\counterwithin*{equation}{part}
\counterwithin*{thm}{part}
\counterwithin*{figure}{part}

\tikzset{meter/.append style={draw, inner sep=10, rectangle, font=\vphantom{A}, minimum width=30, line width=.8, path picture={\draw[black] ([shift={(.1,.3)}]path picture bounding box.south west) to[bend left=50] ([shift={(-.1,.3)}]path picture bounding box.south east);\draw[black,-latex] ([shift={(0,.1)}]path picture bounding box.south) -- ([shift={(.3,-.1)}]path picture bounding box.north);}}}
\tikzset{roundnode/.append style={circle, draw=black, fill=gray!20, thick, minimum size=10mm}}
\tikzset{squarenode/.style={rectangle, draw=black, fill=none, thick, minimum size=10mm}}

\definecolor{Blues5seq1}{RGB}{239,243,255}
\definecolor{Blues5seq2}{RGB}{189,215,231}
\definecolor{Blues5seq3}{RGB}{107,174,214}
\definecolor{Blues5seq4}{RGB}{49,130,189}
\definecolor{Blues5seq5}{RGB}{8,81,156}

\definecolor{Greens5seq1}{RGB}{237,248,233}
\definecolor{Greens5seq2}{RGB}{186,228,179}
\definecolor{Greens5seq3}{RGB}{116,196,118}
\definecolor{Greens5seq4}{RGB}{49,163,84}
\definecolor{Greens5seq5}{RGB}{0,109,44}

\definecolor{Reds5seq1}{RGB}{254,229,217}
\definecolor{Reds5seq2}{RGB}{252,174,145}
\definecolor{Reds5seq3}{RGB}{251,106,74}
\definecolor{Reds5seq4}{RGB}{222,45,38}
\definecolor{Reds5seq5}{RGB}{165,15,21}

\allowdisplaybreaks

\usepackage[most,breakable]{tcolorbox}
\makeatletter
\renewenvironment{boxed}{\begingroup\@ifnextchar\bgroup\boxed@gobegin\boxed@gobegin@empty}{\end{tcolorbox}\endgroup}
\def\boxed@gobegin#1{\def\@tempa{#1}\def\@tempb{orange}\ifx\@tempa\@tempb\begin{tcolorbox}[colback=red!15,colframe=orange!70,breakable,enhanced]\else\begin{tcolorbox}[colback=Blues5seq1,colframe=Blues5seq5,breakable,enhanced]\fi}
\def\boxed@gobegin@empty{\begin{tcolorbox}[colback=Blues5seq1,colframe=Blues5seq5,breakable,enhanced]}
\makeatother

\makeatletter 
\renewcommand\onecolumngrid{
\do@columngrid{one}{\@ne}%
\def\set@footnotewidth{\onecolumngrid}
\def\footnoterule{\kern-6pt\hrule width 1.5in\kern6pt}%
}
\makeatother

\DeclareMathOperator{\stein}{Stein}

\newcommand{\RR}{\pazocal{R}}
\newcommand{\LL}{\pazocal{L}}
\newcommand{\YY}{\pazocal{Y}}

\renewcommand{\SS}{\pazocal{S}}
\renewcommand{\AA}{\pazocal{A}}
\newcommand{\BB}{\pazocal{B}}

\newcommand{\xbar}{\widebar{x}}

\begin{document}

\title{Universal quantum resource distillation\\
via composite 
generalised quantum Stein's lemma}

\author{Ludovico Lami}
\email{ludovico.lami@gmail.com}
\affiliation{Scuola Normale Superiore, Piazza dei Cavalieri 7, 56126 Pisa, Italy}

\author{Bartosz Regula}
\email{bartosz.regula@gmail.com}
\affiliation{Mathematical Quantum Information RIKEN Hakubi Research Team, RIKEN Pioneering Research Institute (PRI) and RIKEN Center for Quantum Computing (RQC), Wako, Saitama 351-0198, Japan}

\author{Ryuji Takagi}
\email{ryujitakagi@g.ecc.u-tokyo.ac.jp}
\affiliation{Department of Basic Science, The University of Tokyo, Tokyo 153-8902, Japan}

\begin{abstract}
The performance of quantum resource manipulation protocols, including key examples such as distillation of quantum entanglement, is measured in terms of the rate at which desired target states can be produced from a given noisy state. However, to achieve optimal rates, known protocols require precise tailoring to the quantum state in question, demanding a perfect knowledge of the input and allowing no errors in its preparation.
Here we show that distillation of quantum resources in the framework of resource non-generating operations can be performed \emph{universally}: optimal rates of distillation can be achieved with no knowledge of the input state whatsoever, certifying the robustness of quantum resource distillation. The findings apply in particular to the purification of quantum entanglement under non-entangling maps, where the optimal rates are governed by the regularised relative entropy of entanglement. Our result relies on an extension of the generalised quantum Stein's lemma in quantum hypothesis testing to a composite setting where the null hypothesis is no longer a fixed quantum state, but is rather composed of i.i.d.\ copies of an unknown state. The solution of this asymptotic problem is made possible through new developments in one-shot quantum information and a refinement of the blurring technique from \href{https://ieeexplore.ieee.org/abstract/document/10898013}{[Lami, IEEE T-IT 71, 4454 (2025)]}. 
\end{abstract}

\maketitle

\section*{Introduction}

\subsection*{Quantum resources and hypothesis testing}

The theory of quantum information is founded on operational tasks. Two particularly important classes of such tasks are quantum resource manipulation, which consists in transforming many copies of a given resourceful object (e.g.\ a quantum state or channel) into as many copies as possible of another object, as well as 
quantum hypothesis testing, which consists in discriminating between alternative assumptions about the nature of an unknown quantum object (typically, a state or a channel).
The performance of such tasks is benchmarked by their optimal asymptotic rates, which can be understood as their average efficiency when more and more copies of a given quantum state or channel become available. 
It is in this context that quantum entropies often emerge and provide a unifying quantitative picture. As it turns out, however, this connection can be much deeper, with the two classes of problems inextricably connected.

The formal treatment of quantum hypothesis testing was initiated by Hiai and Petz~\cite{Hiai1991}, who studied the problem of deciding whether an unknown state over an $n$-copy quantum system is $\rho^{\otimes n}$ or $\sigma^{\otimes n}$, for some known single-system density matrices $\rho$ and $\sigma$. The asymptotic characterisation of the error probabilities involved in this problem in the asymptotic limit $n\to\infty$, known as quantum Stein's lemma, led Hiai and Petz to uncovering the fundamental role played in this context by the quantum relative entropy $D(\rho\|\sigma) \coloneqq \Tr \big[\rho (\log \rho - \log \sigma) \big]$, previously introduced by Umegaki~\cite{Umegaki1962}. 

The paradigmatic example of quantum resource manipulation, on the other hand, is entanglement distillation. It is concerned with transforming copies of a weakly entangled state into as much pure entanglement as possible. The ultimate efficiency of this task, known as distillable entanglement, was first connected with entropic functions by Bennett et al.~\cite{Bennett-distillation}, who showed that the rate at which entanglement of a pure state can be distilled is given by its entropy of entanglement.
Such entropic connections extend in fact to entanglement manipulation problems beyond distillation~\cite{Bennett-error-correction}.
Later generalisations of entropic entanglement measures to mixed states led to the definition of the relative entropy of entanglement~\cite{Vedral1997, Vedral1997-PRA, Vedral1998}, with some connections to hypothesis testing problems already observed in these early works.



This line of work culminated with the formulation of a theory of asymptotic entanglement manipulation based on a fundamental connection with quantum hypothesis testing~\cite{BrandaoPlenio1, BrandaoPlenio2}. This framework, relying on the axiomatically motivated class of asymptotically non-entangling operations (ANE), exhibited a long-sought distinguishing feature of a \emph{unique} measure of entanglement, namely the regularised relative entropy of entanglement $D^\infty(\rho\|\sep)$ based on the Umegaki relative entropy. The central role played by this quantity is encapsulated by the fact that it fully determines all asymptotic transformation rates: for any two states $\rho$ and $\sigma$, the rate of converting the former into the latter is given by $R_{\mathrm{ANE}}(\rho\to \sigma) = \frac{D^\infty(\rho\|\sep)}{D^\infty(\sigma\|\sep)}$. When $\sigma = \Phi_+$ is the maximally entangled state, this gives the distillable entanglement under ANE operations. The same ideas can be extended beyond entanglement, to all quantum resource theories~\cite{Brandao2010,Brandao-Gour} whose corresponding sets of free states~\cite{RT-review}, denoted as $\SS$, obey some basic assumptions, jointly referred to as the `Brand\~ao--Plenio axioms'. These resource theories are all asymptotically reversible under asymptotically resource non-generating operations (ARNG), with the regularised relative entropy of the corresponding resource, denoted $D^\infty(\cdot\|\SS)$, working as the linchpin of the whole theory: similarly to the case of entanglement, all transformation rates can be calculated as $R_{\mathrm{ARNG}}(\rho\to \sigma) = \frac{D^\infty(\rho\|\SS)}{D^\infty(\sigma\|\SS)}$.


The Brand\~ao--Plenio theory of quantum resource manipulation crucially depends on an underlying quantum hypothesis testing problem~\cite{BrandaoPlenio2}, in the sense that the achievability of the above transformation rate relies on a corresponding achievability statement in quantum state discrimination. The problem is the following: given a state $\rho$ on some quantum system and an unknown state $\omega_n$ on $n$ copies of that same system, one needs to find a sequence of tests that optimally discriminates between a null hypothesis $\omega_n = \rho^{\otimes n}$ and an alternative hypothesis $\omega_n\in \SS$ (i.e.\ $\omega_n$ is a free state). Here, `optimally' means that the probability of a type I error (i.e.\ mistaking $\rho^{\otimes n}$ for a free state) can be kept below a given threshold $\e$ while the probability of a type II error (mistaking a free state for $\rho^{\otimes n}$), denoted $\rel{\beta_\e}{\rho^{\otimes n}}{\SS}$, can be made to decay to zero exponentially at the maximal rate given by the regularised relative entropy of the resource: formally, $\rel{\beta_\e}{\rho^{\otimes n}}{\SS} \sim 2^{-n D^\infty(\rho\|\SS)}$. In words, we say that the regularised relative entropy of entanglement represents the Stein exponent of this task; formally, $\stein(\rho\|\SS) = D^\infty(\rho\|\SS)$.
The possibility of constructing the above tests, which directly leads to the asymptotic reversibility of the corresponding quantum resource theory, is known as the generalised quantum Stein's lemma. A proof of this statement was first claimed by Brand\~ao and Plenio themselves~\cite{Brandao2010}, but this was later found to contain a fatal gap~\cite{gap, gap-comment}. The issue was finally settled when two complete solutions of the problem appeared~\cite{Hayashi-Stein, GQSL}.

\subsection*{Summary of results}

From a practical standpoint, a major drawback of this framework is that exact knowledge of the state $\rho$ is required to design the above tests, and, accordingly, also to construct the asymptotic resource non-generating operations that transform copies of $\rho$ into copies of some other target state $\sigma$ at optimal rate. This is physically unrealistic, as perfect knowledge of the source state is never available. Thus, if we intend to apply these results to operationally relevant settings, it is imperative that we demonstrate how to design resource tests and the resulting transformations in a way that is independent of the exact description of underlying source state. 

In this work, we accomplish precisely that. We show, for any given target state $\sigma$, the existence of a \emph{universal} sequence of ARNG protocols that distills it: for any input state $\rho$, without any prior information about the state,\footnote{`Universal' here means that the protocol only depends on the target state $\sigma$ and not on the input state $\rho$; of course, any such protocol must depend on the target state, as the operational problem would not make much sense otherwise.} the protocol can transform copies of $\rho$ into copies of $\sigma$ at an optimal rate of
\bb
R_{\mathrm{ARNG}}^{\mathrm{u}}(\rho\to \sigma) = \frac{D^\infty(\rho\|\SS)}{D^\infty(\sigma\|\SS)}\, .
\ee
A remarkable feature of the result is that it establishes that the best-case scenario is realised. Indeed, the highest transformation rate that we could hope for, even with a full knowledge of the input state $\rho$, is precisely $\frac{D^\infty(\rho\|\SS)}{D^\infty(\sigma\|\SS)}$; our contribution is to prove that this bound is achieved even with no prior information about the state whatsoever.
In the particularly important case of entanglement distillation, the result applies with the choice of non-entangling operations (NE), showing that there exists a universal entanglement distillation protocol that converts any noisy entangled state $\rho=\rho_{AB}$ into singlets $\Phi_+$ at a rate given by
\bb
R_{\mathrm{NE}}^{\mathrm{u}}(\rho\to \Phi_+) = D^\infty(\rho\|\sep)\, .
\ee

Variants of such universal protocols were previously studied in other quantum information processing tasks~\cite{Jozsa1998,hayashi_2002-3,Josza2003,bennett_2006,hayashi_2009-2,bjelakovic_2013,matsuura2025asymptoticallytightsecurityanalysis,matsuura2025universalclassicalquantumchannelresolvability} and in the limited case of pure-state entanglement transformations~\cite{matsumoto_2007}. More recently, universal resource manipulation protocols have been introduced in the resource theory of thermodynamics~\cite{watanabe_2024,watanabe_2026,faist_2026}. 
However, they have not previously been applied to more general resource theories such as quantum entanglement of arbitrary mixed states, with a major technical obstruction being an insufficient generality of the known results in quantum hypothesis testing --- the tests required for universal resource transformations cannot be directly obtained from the methods used in the original conjecture of~\cite{Brandao2010}, or from the recent solutions of the generalised quantum Stein's lemma~\cite{Hayashi-Stein, GQSL}, or from the subsequent advances in composite quantum hypothesis testing spurred by these solutions~\cite{generalised-Sanov,Hayashi-Sanov-2,Fang2025,fang_2025-1,doubly-comp-quantum}.


To establish the result, we thus tackle and resolve a class of hypothesis testing problems that extends the generalised quantum Stein's lemma. To understand the difficulty in constructing resource transformations in this way, consider a natural way to perform distillation from a completely unknown state: first, collect some statistics of the unknown quantum state in order to obtain its estimate; then, try to use the Brand\~ao--Plenio framework~\cite{BrandaoPlenio2} to perform an asymptotically non-entangling transformation based on the estimate. The key problem is that performing full state tomography would use up a large number of copies of the source state, hindering the efficiency of the transformation itself. An effort to avoid a loss of performance confines us to a weaker estimate of the state --- say, $\rho \in \RR_1$, where the set $\RR_1$ can be understood as a set of quantum states that are compatible with our experimental data on the source, obtained by sacrificing a small fraction of it.
We are then faced the following hypothesis testing problem: given an unknown state $\omega_n$ over $n$ copies of a quantum system, distinguish between 
\begin{itemize}[leftmargin=10pt]
\item Null hypothesis: $\omega_n \in \RR_n^\iid$, which means that $\omega_n = \rho^{\otimes n}$ for some $\rho\in \RR_1$,
\item Alternative hypothesis: $\omega_n\in \SS$, where $\SS$ is a sequence of sets of free states assumed to satisfy the Brand\~ao--Plenio axioms.
\end{itemize}
Here we treat the set $\RR_1$ generally and leave it completely arbitrary. Our main result states that, for all numbers $\e\in (0,1)$, it is possible to construct a sequence of universal tests that discriminate between the two above hypotheses in such a way that:
\begin{itemize}[leftmargin=10pt]
\item the type I error probability is kept below $\e$ for all $n$;
\item the type II error probability, denoted $\rel{\beta_\e}{\RR_n^{\iid}}{\SS}$, decays exponentially to zero as
\bb
\rel{\beta_\e}{\RR_n^{\iid}}{\SS} &\sim 2^{-n \stein\left(\RR^\iid\middle\|\SS\right)},\\
\stein(\RR^\iid\|\SS) &= \inf_{\rho\in \RR_1} D^\infty(\rho\|\SS)\, .
\ee
\end{itemize}
In order to establish this result, we refine the technique of `blurring', originally introduced in~\cite{GQSL} for the resolution of the simple i.i.d.\ form of the same problem and already successfully used in other important hypothesis testing problems~\cite{generalised-Sanov,doubly-comp-quantum}.  Although seemingly insufficient to tackle this composite variant in its basic form, we show that a suitable technical extensions can enhance the applicability of the blurring method. 



\section*{Quantum hypothesis testing}

Let $\HH$ be a finite-dimensional Hilbert space. Quantum states on $\HH$ are represented by density matrices, i.e.\ positive semi-definite matrices with unit trace. We will denote with $\D(\HH)$ the set of density operators on $\HH$. The fundamental task of composite quantum hypothesis testing is defined by two sequences of hypotheses $\RR = (\RR_n)_n$ and $\SS = (\SS_n)_n$, where each $\RR_n, \SS_n\subseteq \D(\HH^{\otimes n})$ is a set of quantum states on $\HH^{\otimes n}$. 
Operationally, the problem of hypothesis testing can be phrased as follows. One is given an unknown state $\omega_n\in \D\big(\HH^{\otimes n}\big)$, with the promise that one of the following two hypotheses is true: either the null hypothesis $\omega_n \in \RR_n$, or the alternative hypothesis $\omega_n\in \SS_n$. 
The goal is to distinguish between these two options, using an appropriate quantum measurement --- which, in general, will be over the whole $n$-copy system. 
There are then two types of error: a type I error, which consists in guessing the alternative hypothesis when the null hypothesis holds, and a type II error, which, conversely, consists in guessing the null hypothesis when the alternative hypothesis holds. The minimum type II error at a fixed type I error threshold, for a given $n$, can be expressed as
\bb
\beta_\e(\RR_n\|\SS_n) \coloneqq \inf&\left\{ \sup_{\sigma_n\in \SS_n} \Tr \sigma_n E_n:\ 0\leq E_n\leq \id, \right.\\
&\left.\phantom{\sup_{\sigma_n\in \SS_n} \Tr } \inf_{\rho_n\in \RR_n} \Tr \rho_n E_n \geq 1-\e \right\} ,
\ee
where $E_n$ is understood to act on $\HH^{\otimes n}$. An application of Sion's minimax theorem shows that we can alternatively rewrite (see e.g.~\cite[Lemma~31]{Fang2025})
\bb
- \log \beta_\e(\RR_n\|\SS_n) &\hphantom{:}= \rel{D_H^\e}{\co(\RR_n)}{\co(\SS_n)} 
\\&\coloneqq \!\inf_{\substack{\rho_n \in \co(\RR_n),\\ \sigma_n\in \co(\SS_n)}} D_H^\e(\rho_n\|\sigma_n)\, ,
\ee
where 
\bb
D_H^\e(\rho\|\sigma) \!\coloneqq\! -\log \min\left\{ \Tr \sigma E: 0\leq E\leq \id,\ \Tr \rho E\geq 1-\e\right\}
\label{D_H}
\ee
is the so-called hypothesis testing relative entropy~\cite{Buscemi2010}. 

The \deff{Stein exponent} characterises the asymptotic decay of $\beta_\e(\RR_n\|\SS_n)$ as a function of $n$:
\bb
\stein(\RR\|\SS) \!\coloneqq\! \lim_{\e\to 0^+} \liminf_{n\to\infty} \frac1n\, \rel{D_H^\e}{\co(\RR_n)}{\co(\SS_n)}\, ,
\label{Stein}
\ee
where, given a divergence $\mathds{D}(\cdot\|\cdot)$ and two sets $\AA$, $\BB$, we set $\mathds{D}(\AA\|\BB) \coloneqq \inf_{\rho\in \AA,\, \sigma\in \BB} \mathds{D}(\rho\|\sigma)$. 
We can also define a \deff{strong converse Stein exponent}, given by
\bb
\stein^\dag(\RR\|\SS) \!\coloneqq\! \lim_{\e\to 1^-} \limsup_{n\to\infty} \frac1n\, \rel{D_H^\e}{\co(\RR_n)}{\co(\SS_n)}\, .
\label{strong_converse_Stein}
\ee

In this paper, we will deal with the case where $\RR$ is a \deff{composite i.i.d.\ hypothesis}, meaning that
\bb
\RR = \RR^\iid \coloneqq \big( \RR_n^{\iid} \big)_n\, ,\quad \RR_n^{\iid} \coloneqq \left\{ \rho^{\otimes n}:\ \rho\in \RR_1\right\}
\label{composite_iid}
\ee
for some arbitrary single-copy base set $\RR_1\subseteq \D(\HH)$, while the alternative hypothesis is built from potentially much more general sets $\SS_n$, exhibiting a structure that could go far beyond i.i.d. 
As discussed in the Introduction, the motivation to consider such a setting comes from the problem of universal resource transformations in quantum resource theories: there, $\SS_n$ corresponds to the free states of a given theory, while tomographic estimates of unknown states make it necessary to understand composite i.i.d.\ hypothesis testing where $\RR_1$ is a small neighbourhood around some central state. Before being able to address that problem, we must therefore tackle the underlying hypothesis testing task.

One of the first extensions of the quantum Stein's lemma~\cite{Hiai1991} beyond the case of two simple i.i.d.\ hypotheses was a `quantum Sanov's theorem' of Bjelakovic et al.~\cite{Bjelakovic2005}, which considered precisely the case of a composite i.i.d.\ hypothesis $\RR^\iid$, while keeping the alternative hypothesis as the simple i.i.d.\ sequence $\SS_n = \{ \sigma^{\otimes n} \}$. 
What one observes empirically in such composite i.i.d.\ extensions is that the Stein exponent is given by the infimum, i.e.\ the worst case, of the Stein exponents corresponding to the cases where the null hypothesis is simple and i.i.d. Formally, what one could expect is that
\bb
\rel{\stein}{\RR^\iid}{\SS} \eqt{?} \inf_{\rho\in \RR_1} \stein(\rho\|\SS)\, ,
\label{goal_Stein_reduction_composite_iid}
\ee
where on the right-hand side we identified $\rho$ with the sequence of its i.i.d.\ repetitions.
However, because we are here concerned with a fully composite alternative hypothesis $\SS$, the previous methods do not apply --- and indeed, even for simple i.i.d.\ null hypotheses, the composite structure of the alternative one can cause considerable difficulty~\cite{Brandao2010,gap,Hayashi-Stein, GQSL}.

\section*{Main result \#1:\\composite i.i.d.\ hypothesis testing}

Our main result in the context of quantum hypothesis testing states that the conjectured composite i.i.d.\ generalised quantum Stein's lemma in Eq.~\eqref{goal_Stein_reduction_composite_iid} does indeed hold in the physically important case where $\SS$ satisfies a set of natural properties which mirror the basic structural properties of the set of separable states $\SEP$. These are the Brand\~ao--Plenio axioms~\cite{Brandao2010}: (1) convexity and topological closedness; (2) existence of a full-rank free state; (3) closedness under tracing out subsystems; (4) closedness under taking tensor products of free states; (5) closedness under permutations of subsystems.

\begin{boxed}
\begin{thm} \label{black_box_Stein_iid_thm}
Let $\RR_1\subseteq \D(\HH)$ be any set of states on the finite-dimensional Hilbert space $\HH$. Assume that the sequence $\SS = (\SS_n)_n$ of alternative hypotheses $\SS_n\subseteq \D(\HH^{\otimes n})$ satisfies all five Brand\~{a}o--Plenio axioms. Then
\bb
\lim_{n\to\infty} \frac1n\, \rel{D_H^\e}{\co\!\big(\RR_n^{\iid}\big)}{\SS_n} = \inf_{\rho\in \RR_1} D^\infty(\rho \|\SS)
\label{black_box_Stein_iid_DH}
\ee
for all $\e\in (0,1)$, 
where we recall that  $\RR_n^{\iid} = \left\{ \rho^{\otimes n}:\ \rho\in \RR_1\right\}$.
Equivalently, denoting by $\RR^\iid$ the sequence of sets $\big(\RR_n^{\iid}\big)_n$, we have
\bb
&\rel{\stein}{\RR^{\mathrm{iid}}}{\SS} = \rel{\stein^\dag}{\RR^{\mathrm{iid}}}{\SS} \\
&= \inf_{\rho\in \RR_1} D^\infty(\rho\|\SS) = \inf_{\rho\in \RR_1} \stein(\rho\|\SS)\, .
\label{black_box_Stein_iid}
\ee
\end{thm}
\end{boxed}

At first glance, the problem at hand could seem a technically immediate extension of the standard generalised quantum Stein's lemma, which states that~\cite{Hayashi-Stein, GQSL}
\bb
\lim_{n\to\infty} \frac1n\, \rel{D_H^\e}{\rho^{\otimes n}}{\SS_n} = D^\infty(\rho \|\SS)
\label{GQSL}
\ee
for all fixed states $\rho$ and all $\e\in (0,1)$. However, this is not the case: the crucial difference is that in~\eqref{black_box_Stein_iid_DH} we have not merely an i.i.d.\ state from some set, but in fact a generic convex combination of multiple i.i.d.\ states. 
Physically, this precisely corresponds to the difference between designing a test that is tailored to a fixed state $\rho$ in the null hypothesis and designing a test that works universally for \emph{all} states $\rho\in \RR_1$. 

The first step to overcome this hurdle is to remove the need to consider convex combinations and go back to testing against a \emph{sequence} of i.i.d.\ states $\rho(n)^{\otimes n}$, where, for each $n$, $\rho(n)\in \RR_1$ is a state in the base set. We accomplish this through a novel quasi-concavity property for the hypothesis testing relative entropy, which allows us to do this at the expense of introducing small error terms that become insignificant in the asymptotic limit. 
This technical development relies on the formalism of one-shot relative entropies, and in particular on the `weak/strong converse' duality between hypothesis testing and the smooth max-relative entropy~\cite{Tomamichel2013, Anshu2019, tight-relations}. However, this on its own does not suffice: indeed, while $\rho(n)^{\otimes n}$ is i.i.d.\ for every $n$, the base state $\rho(n)$ is not fixed, and can in general `fluctuate' around some limiting state $\widebar{\rho} = \lim_n \rho(n)$. The known formulations of the generalised quantum Stein's lemma~\cite{Hayashi-Stein, GQSL} are not robust against such fluctuations.

To conclude, we thus need a strengthened variant of the generalised quantum Stein's lemma that would allow us to cover this setting. To obtain it, we go back to the blurring technique devised in~\cite{GQSL}. This is based on the idea that adding a carefully controlled amount of noise to an unknown state can `tame' it enough to make it amenable to analysis. The blurring map acts on some $n$-partite input operator $X_n$ as
\bb
\widebar{B}_{n,\delta} (X_n) &\coloneqq \Tr_{\floor{n \delta}} \E_{\pi}\, U_\pi^{\vphantom{\dag}}\Big(X_n \otimes \sigma_0^{\otimes \floor{n \delta}} \Big) U_\pi^\dag\, ,
\ee
where $\delta>0$ is a small parameter that quantifies how many subsystems in some noisy state $\sigma_0$ we add, $\pi$ is a uniformly random permutation that shuffles the subsystems, and the final partial trace brings back the total number of systems to $n$.

The `quantum blurring lemma' of~\cite{GQSL}, however, is in itself not strong enough to handle the case of a fluctuating i.i.d.\ state that is of interest here. And yet, intuitively, since our fluctuating state $\rho(n)$ is very close to $\widebar{\rho}$ in the limit of large $n$, the small amount of noise introduced by the blurring map should nevertheless succeed in making it essentially identical to $\widebar{\rho}$ for our purposes. 

To analyse this problem mathematically and show that this is indeed what happens, we devise a dynamical version of the `bosonic lifting' technique of~\cite{GQSL}. In itself, bosonic lifting is the procedure of embedding our sequence of states into a single infinite-dimensional, bosonic space, where it can be analysed with techniques from continuous-variable quantum information theory. In this picture, the reference state is mapped into the vacuum state of the bosonic system. In our case, we have a fluctuating reference state that depends on $n$, and this means that the embedding needs to become itself $n$-dependent. With a careful mathematical analysis --- the complete details of which we provide in the \hyperlink{supp}{Appendix} --- we thus adapt quantum blurring and the bosonic lifting technique to handle fluctuating sources, finally proving Theorem~\ref{black_box_Stein_iid_thm}.

\section*{Main result \#2:\\universal transformations of quantum resources}

To explain the details of the connection between hypothesis testing and resource transformations formalised in Brand\~ao and Plenio's framework~\cite{BrandaoPlenio2}, let us consider the illustrative example of entanglement distillation. Here we are concerned with the conversion of a fixed, known input state $\rho = \rho_{AB}$ into copies of the maximally entangled state $\Phi_+ = \ketbra{\Phi_+}$. For a given rate $R$, define then the sequence of channels
\begin{equation}\begin{aligned}
  \EE_n(\rho^{\otimes n}) =  \Tr \left[ E_n \rho^{\otimes n} \right] \Phi_+ + \Tr \left[ (\id - E_n) \rho^{\otimes n} \right] \frac{\id - \Phi_+^{\otimes \floor{Rn}}}{4^{\floor{R n}} - 1}.
\end{aligned}\end{equation}
Two key facts can be observed here. First, the fidelity of the output state $\EE_n(\rho^{\otimes n})$ to $\Phi^{\otimes \floor{Rn}}$ equals $\Tr \left[ E_n \rho^{\otimes n} \right]$, which is precisely the type I error under the test $(E_n, \id - E_n)$. Second, if $\Tr \left[ E_n \omega_n \right] \leq 2^{- R n}$ for some state $\omega_n$, then $\EE_n(\omega_n) \in \SEP_n$, as can be readily verified thanks to the isotropic structure of the output state~\cite{Horodecki-teleportation}. What this means is that, for any rate $R$ such that $2^{-R n} = \beta_\ve (\rho^{\otimes n} \| \SEP_n)$, the map $\EE_n$ constructed from the optimal hypothesis test is non-entangling, i.e.\ it always maps separable states into separable states, and crucially it performs the conversion of $\rho^{\otimes n}$ into $\Phi_+^{\otimes \floor{Rn}}$ with an error of $\ve$. Denoting by $R_{\rm NE}(\rho \to \Phi_+)$ the best achievable rate, i.e.\ the supremum of all rates $R$ achievable with error $\ve$ that vanishes asymptotically, the above tells us that this rate is exactly $R_{\rm NE}(\rho \to \Phi_+) = \stein(\rho \| \SEP)$. By the generalised quantum Stein's lemma~\cite{Brandao2010,Hayashi-Stein, GQSL}, this equals precisely $D^\infty(\rho\|\SEP)$. 
A more general variant of this operational connection was shown in~\cite{BrandaoPlenio2} for arbitrary target states $\sigma$, provided that the choice of non-entangling operations (NE) is replaced with asymptotically non-entangling ones (ANE). As already shown in~\cite{Brandao2010, Hayashi-Stein, GQSL}, whenever $\SS$ satisfies the Brand\~ao--Plenio axioms, the equivalence $R_{\rm ARNG}(\rho \to \sigma) = D^\infty(\rho\|\SS)$ holds~\cite{Brandao-Gour,Hayashi-Stein, GQSL}. 

However, we see also from this construction that the design of the protocol is intrinsically based on the knowledge of the input state $\rho$, as the optimal test must be chosen for this specific state.
The main application of our composite i.i.d.\ extension of the generalised quantum Stein's lemma is to enable a construction of protocols that do not require any knowledge of the input state. 
Formally, we will say that a rate $R$ is universally $\ve$-achievable with asymptotically resource--non-generating operations if there exists a sequence of channels $(\EE_n)_n \in \rm{ARNG}$ such that $\frac12 \| \EE_n(\rho^{\otimes n}) - \omega^{\otimes Rn}\|_1 \leq \ve$ for any state $\rho^{\otimes n}$. The best rate achievable with vanishing error is then
\begin{equation}\begin{aligned}
  R_{\mathrm{ARNG}}^{\mathrm{u}}(\rho\to \omega) \coloneqq \lim_{\ve \to 0^+} \sup \left\{ R : R \text{ is $\ve$-achievable} \right\}.
\end{aligned}\end{equation}
Clearly, $R_{\mathrm{ARNG}}^{\mathrm{u}}(\rho\to \omega)$ is always upper bounded by the rate $R_{\mathrm{ARNG}}(\rho\to \omega)$, which does not require the feasible transformation protocols $\EE_n$ to be universal. Our main result is to show that the two are in fact equal, and the optimal rate of $D^\infty(\rho\|\SS)$ is achieved in a completely universal fashion.

\begin{boxed}
\begin{thm}
Let $\HH$ be a finite-dimensional Hilbert space, and let $\SS = (\SS_n)_n$ be a sequence of sets of $n$-copy states $\SS_n\subseteq \D(\HH^{\otimes n})$ that satisfies all five Brand\~{a}o--Plenio axioms. Further, assume that there exists a tomographically complete set of measurements that is compatible with $\SS$, and pick some $\omega\in \D(\HH)$ such that $D^\infty(\omega\|\SS) > 0$. Then, for all $\delta>0$, there exists a universal protocol that uses only asymptotically resource non-generating operations and distills copies of $\omega$ from copies of any unknown state $\rho$ with asymptotically vanishing error and rate
\bb
R \geq \frac{D^\infty(\rho\|\SS)}{D^\infty(\omega\|\SS)} - \delta\,.
\ee
Hence, for any $\rho$, the optimal rate achievable universally equals the optimal rate achievable with full knowledge of the state $\rho$,
\bb
R_{\mathrm{ARNG}}^{\mathrm{u}}(\rho\to \omega) = R_{\mathrm{ARNG}}(\rho\to \omega) = \frac{D^\infty(\rho\|\SS)}{D^\infty(\omega\|\SS)}.
\ee

For the particular case of entanglement distillation, this result applies under non-entangling operations and yields
\bb
R_{\mathrm{NE}}^{\mathrm{u}}(\rho\to \Phi_+) = D^\infty(\rho\|\SEP)\,.
\ee
\end{thm}
\end{boxed}

In the above theorem, the requirement of tomographical completeness of a compatible set of measurements is similar to assumptions typically appearing in related settings~\cite{brandao_adversarial,generalised-Sanov,Hayashi-Sanov-2}. What it means is that there exists a class of measurements with the following two properties: first, their combined statistics suffice to determine the state; and secondly, measuring a subsystem of a quantum system that is in a free state leaves a free state in the remaining subsystems, for every admissible outcome. This property is satisfied in all quantum resources known to us and is merely a technical requirement.

\begin{proof}
Let $\rho$ be any state --- unknown to us --- and let $B_\e(\rho)$ denote the set of states at trace norm distance at most $\e$ from $\rho$. 
For some $\mu>0$ to be determined later, we consider the following protocol: first, sacrifice $\floor{\mu n}$ copies of the unknown state $\rho$ to run tomography with the compatible set of measurements, whose existence is guaranteed by assumption. This will produce an ansatz $\widebar{\rho}$ together with the guarantee that $\widebar\rho \in B_\e(\rho)$ with probability $p_n \tendsn{} 1$. On the remaining copies of the state, of which there are at least $(1-\mu) n$, we can then apply the Brand\~ao--Plenio protocol associated with the hypothesis testing task of distinguishing $B_\e\big(\,\widebar{\rho}\,\big)^\iid$ from $\SS$. In itself, this can be made to distil copies of $\omega$ at an asymptotic rate 
\bb
R' \geq \frac{\Rel{\stein\!}{B_\e\big(\,\widebar{\rho}\,\big)^\iid}{\SS} - \xi}{D^\infty(\omega\|\SS)}
\ee
with asymptotically vanishing error; here, $\xi>0$ is an arbitrarily small constant. The overall rate we are achieving is smaller than this, however, because we are running the distillation protocol only on the remaining $(1-\mu)n$ copies of the unknown state. This means that in fact
\bb\label{eq:univ_rate_achiev}
R_{\mathrm{RNG}}^{\mathrm{u}}(\rho\to \omega) &\geq (1-\mu) R'\\
&\geq (1-\mu)\, \frac{\Rel{\stein\!}{B_\e\big(\,\widebar{\rho}\,\big)^\iid}{\SS}- \xi}{D^\infty(\omega\|\SS)}\, .
\ee
Invoking Theorem~\ref{black_box_Stein_iid_thm} for the choice of the set $B_\e(\rho)$ tells us that
\bb
\lim_{\e\to 0^+} \rel{\stein}{B_\e(\rho)^\iid}{\SS} = \stein(\rho\|\SS) = D^\infty(\rho\|\SS)\, .
\ee
Therefore, the rightmost side of~\eqref{eq:univ_rate_achiev} can be made larger than $D^\infty(\rho\|\SS) / D^\infty(\omega\|\SS) - \delta$ for any $\delta$ by choosing $\mu$, $\xi$, and $\e$ small enough. 

Due to the fact that tomography is performed using measurements that are compatible with $\SS$, and the Brand\~ao--Plenio protocol is asymptotically resource non-generating, the overall protocol is also asymptotically resource non-generating. 
This can be seen as follows. 
Let $\{M_k\}_k$ be a POVM on a single subsystem representing the tomographically complete measurement compatible with $\SS$. The resource distillation protocol is run by applying the tomographically complete measurement on $\mu n$ copies, represented by $\{M_{\vec{k}}\}_{\vec{k}}$ where $M_{\vec{k}} = M_{k_1}\otimes M_{k_2}\otimes \dots \otimes M_{k_{\mu n}}$ is an POVM element with output $\vec{k}=(k_1,\dots, k_{\mu n})$ of tomographic data, followed by applying an asymptotically resource non-generating operation $\{\Lambda_{(1-\mu)n}^{\vec{k}}\}_n$ depending on the tomographic data $\vec{k}$ to the $(1-\mu)n$ remaining systems. 
This means that the protocol is represented by the channel $\sum_{\vec{k}}\Lambda_{(1-\mu)n}^{\vec{k}}\otimes \pazocal{M}_{\vec{k}}$ where $\pazocal{M}_{\vec{k}}(\cdot) = \Tr(M_{\vec{k}}\,\cdot)$.
Recall now that the fact that $\{\Lambda_{(1-\mu)n}\}_n$ is asymptotically resource non-generating ensures that there exists a series $\{\delta_n\}_n$ of positive real numbers such that $R_g(\Lambda_{(1-\mu)n}(\sigma_n))\leq \delta_n$ and $\delta_n\xrightarrow[n\to\infty]{}0$ for an arbitrary sequence $\{\sigma_n\}_n$ of free states, where $R_g(\rho)=\min\{s\geq 1\,|\, \rho\leq s\sigma, \sigma\in\SS\}$ is the generalized robustness of resource.
Then, we get for any free state $\sigma_n\in\SS$ that
\bb
R_g\left(\sum_k\Lambda_{(1-\mu)n}^{\vec{k}}\otimes \pazocal{M}_{\vec{k}}(\sigma_n)\right) &= R_g\left(\sum_{\vec{k}}p_{\vec{k}}\Lambda_{(1-\mu)n}^{\vec{k}}(\sigma_{(1-\mu)n}^{\vec{k}})\right)\\
& \leq  \sum_{\vec{k}}p_{\vec{k}}R_g\left(\Lambda_{(1-\mu)n}^{\vec{k}}(\sigma_{(1-\mu)n}^{\vec{k}})\right)\\
&\leq \sum_{\vec{k}}p_{\vec{k}} \delta_{(1-\mu)n} \xrightarrow[n\to\infty]{} 0
\ee
where we set $p_{\vec{k}}=\Tr(M_{\vec{k}}\sigma_n)$ and $\sigma^{\vec{k}}_{(1-\mu)n}=\Tr_{\mu n}(\mathbb{I}\otimes M_{\vec{k}}\,\sigma_n)/p_{\vec{k}}$ is the post-measurement state conditioned on the outcome $\vec{k}$ of the measurement applied to the first $\mu n$ subsystems. 
In the second line, we used the convexity of the generalized robustness. In the third line, we used that $\sigma^{\vec{k}}_{(1-\mu)n}\in\SS$ for every $\vec{k}$ because of the assumption that $\{M_{\vec{k}}\}_{\vec{k}}$ is compatible with $\SS$.
This ensures that the distillation protocol as a whole is also asymptotically resource non-generating.
Our achievability proof is thus complete.

The argument for entanglement distillation proceeds analogously, using now the distillation protocol of~\cite{BrandaoPlenio2}, in which case the class of non-entangling operations suffices. Using the fact that $R_{\mathrm{ARNG}}^{\mathrm{u}}(\rho\to \omega) \leq R_{\mathrm{ARNG}}(\rho\to \omega) \leq D^\infty(\rho\|\SS)$~\cite{BrandaoPlenio2,gap} concludes the proof of the theorem.
\end{proof}

\acknowledgments
We are grateful to Hayata Yamasaki and Marco Tomamichel for discussions. L.L.\ acknowledges financial support from the European Union (ERC StG ETQO, Grant Agreement no.\ 101165230). B.R.\ was supported by the Japan Science and Technology Agency (JST) PRESTO grant no.\ JPMJPR25FB. R.T.\ was supported by JST CREST Grant Number JPMJCR23I3, JSPS KAKENHI Grant Number JP24K16975, JP25K00924, JP26H02015, JST NEXUS Grant Number JPMJNX26C2.

\bibliography{biblio}

\clearpage
\onecolumngrid
\appendix 

\renewcommand{\theequation}{S\arabic{equation}}
\setcounter{equation}{0}

\hypertarget{supp}{}
\begin{center}
\vspace*{\baselineskip}
{\textbf{\large --- Appendix ---}}\\[1pt] \quad \\
\end{center}


The Appendix is devoted to a complete proof of the composite i.i.d.\ generalised quantum Stein's lemma in Theorem~\ref{black_box_Stein_iid_thm}. Before we tackle the main proof through an extension of blurring techniques in Appendix~\ref{sec:mainproof}, we begin in Appendix~\ref{sec_quasi_concavility_smooth_relent} with a collection of useful technical results that will then allow us to reduce the problem of testing convex combinations of i.i.d.\ states to the testing of a single one in Appendix~\ref{sec:reduction}. 

\section{Approximate quasi-concavity of one-shot relative entropies} \label{sec_quasi_concavility_smooth_relent}

Following~\cite{tight-relations}, we will write $D_{\max}^{\e,P}$ for the smooth max-relative entropy, where the smoothing is with respect to the purified distance $P(\rho,\sigma)\coloneqq \sqrt{1 - F^2(\rho,\sigma)}$, with $F(\rho,\sigma) \coloneqq \big\|\sqrt{\rho}\sqrt{\sigma} \big\|_1$. Without loss of generality, we smooth only over normalised states.

\begin{lemma} \label{decomposition_lemma}
Let $(p(x),\rho_x)_x$ be a finite ensemble of states on the same quantum system, and let $\omega$ denote another state on the same system. Then there exists a decomposition $\omega = \sum_x q(x)\,\omega_x$ such that
\bb
\sum_x \sqrt{p(x)\, q(x)}\, F(\rho_x,\omega_x) = F\left(\sumno_x p(x)\,\rho_x,\, \omega \right) .
\ee
\end{lemma}

\begin{proof}
The fact that the left-hand side is at most equal to the right-hand side for \emph{any} decomposition $\omega = \sum_x q(x)\,\omega_x$ is a simple application of the strong concavity of the fidelity. See, for example,~\cite[Theorem~9.7]{NC}. The opposite inequality relies on Uhlmann's theorem, and its proof follows closely that of~\cite[Theorem~9.7]{NC}. Indeed, denoting with $A$ the underlying quantum systems, for all $x$ call $\ket{\psi_x}$ a purification of $\rho_x$ on $AA'$, where, without loss of generality, $A'\simeq A$ is the same for all $x$. Construct the associated purification
\bb
\ket{\psi}_{AA'X} \coloneqq \sum_x \sqrt{p(x)} \ket{\psi_x}_{AA'} \ket{x}_X
\ee
of $\sum_x p(x)\,\rho_x$. By Uhlmann's theorem, there is a purification 
\bb
\ket{\phi}_{AA'X} = \sum_x \sqrt{q(x)} \ket{\phi_x}_{AA'} \ket{x}_X
\ee
of $\omega$ such that
\bb
F\left(\sumno_x p(x)\,\rho_x,\, \omega\right) &= F(\psi,\phi) \\
&\leqt{(i)} F\left( \sumno_x p(x)\, \psi_x^{AA'} \otimes \ketbra{x}_X,\ \sumno_x q(x)\, \phi_x^{AA'} \otimes \ketbra{x}_X \right) \\
&= \sum_x \sqrt{p(x)\, q(x)}\, F(\psi_x,\phi_x) \\
&\leqt{(ii)} \sum_x \sqrt{p(x)\, q(x)}\, F(\rho_x,\omega_x)\, .
\ee
Here, inequalities~(i) and~(ii) hold both by data processing: in~(i) we dephased system $X$, while in~(ii) we traced away $A'$, defining the states $\omega_x \coloneqq \Tr_{A'} \phi_x^{AA'}$.
\end{proof}

The following result shows that when the first argument of the smooth max-relative entropy is a convex combination of states, we can in fact approximate its value by considering only the worst-case element of the convex combination, incurring only an additive penalty that does not affect the asymptotic behaviour of the function.

\begin{lemma}[(Approximate quasi-concavity of the smooth max-relative entropy)] \label{quasi_concavity_D_max_eps_thm}
Let $(p(x),\rho_x)_{x\in \XX}$ be a finite ensemble of states on the same quantum system, and let $\sigma$ be another state on the same system. Then, for all $\e,\mu\in (0,1)$ with $\e+\mu\leq 1$, it holds that
\begin{align}
\min_x D_{\max}^{\e+\mu,\,P}\!(\rho_x\|\sigma) &\leq \Rel{D_{\max}^{\e,P}}{\sumno_{x\in \XX} p(x)\, \rho_x}{\sigma} + \log \frac{|\XX|}{(\e+\mu)^2 - \e^2}\, . \label{quasi_concavity_D_max_eps}
\end{align}
\end{lemma}

\begin{proof}
Set $\rho\coloneqq \sum_{x\in \XX} p(x)\,\rho_x$, and let $\omega$ be a state such that 
\bb
P(\rho,\omega) = \e\, ,\qquad D_{\max}(\omega\|\sigma) = D_{\max}^{\e,P} (\rho\|\sigma)\, .
\label{quasi_concavity_D_max_eps_proof_eq0}
\ee
By Lemma~\ref{decomposition_lemma}, there exists a convex decomposition of $\omega$ as
\bb
\omega = \sum_{x\in \XX} q(x)\,\omega_x
\label{quasi_concavity_D_max_eps_proof_eq0_bis}
\ee
such that
\bb
\sum_{x\in \XX} \sqrt{p(x)\, q(x)}\, F(\rho_x,\omega_x) = F\left(\rho, \omega \right) = \sqrt{1 - P(\rho,\omega)^2} = \sqrt{1 - \e^2}\, .
\ee
Now, for some $\delta\in \Big[0,\sqrt{1-\e^2}\Big)$ to be fixed later, define the set
\bb
\XX_\delta \coloneqq \left\{ x\in \XX:\ q(x)\geq\frac{\delta^2}{|\XX|} \right\} .
\label{quasi_concavity_D_max_eps_proof_eq2}
\ee
We can now write
\bb
\sqrt{1 - \e^2} &= \sum_{x\in \XX} \sqrt{p(x)\, q(x)}\, F(\rho_x,\omega_x) \\
&= \sum_{x\in \XX_\delta} \sqrt{p(x)\, q(x)}\, F(\rho_x,\omega_x) + \sum_{x\in \XX_\delta^c} \sqrt{p(x)\, q(x)}\, F(\rho_x,\omega_x) \\
&\leqt{(i)} \sum_{x\in \XX_\delta} \sqrt{p(x)\, q(x)}\, F(\rho_x,\omega_x) + \sum_{x\in \XX_\delta^c} \sqrt{p(x)\, q(x)} \\
&\leqt{(ii)} \sum_{x\in \XX_\delta} \sqrt{p(x)\, q(x)}\, F(\rho_x,\omega_x) + \delta \sqrt{\sumno_{x\in \XX_\delta^c} p(x)} \\
&\eqt{(iii)} \sum_{x\in \XX_\delta} \sqrt{p(x)\, q(x)}\, F(\rho_x,\omega_x) + \delta \sqrt{1-p(\XX_\delta)} \\
&\eqt{(iv)} \sqrt{p(\XX_\delta)} \sum_{x\in \XX_\delta} \sqrt{\frac{p(x)\, q(x)}{p(\XX_\delta)}}\, F(\rho_x,\omega_x) + \delta \sqrt{1-p(\XX_\delta)}\, .
\label{quasi_concavity_D_max_eps_proof_eq3}
\ee
Here, in~(i) we upper bounded $F(\rho_x,\omega_x)\leq 1$ for all $x\in \XX_\delta^c$, in~(ii) we used the Cauchy--Schwarz inequality, and in~(iii) we defined 
\bb
p(\XX_\delta) \coloneqq \sum_{x\in \XX_\delta} p(x) = 1 - \sum_{x\in \XX_\delta^c} p(x)\, ,
\label{quasi_concavity_D_max_eps_proof_eq4}
\ee
observing in~(iv) that we need to have $p(\XX_\delta)>0$, as $p(\XX_\delta) = 0$ would imply that $\delta \geq \sqrt{1-\e^2}$, contrary to our assumption.

Continuing, from~\eqref{quasi_concavity_D_max_eps_proof_eq3} we obtain that
\bb
\sum_{x\in \XX_\delta} \sqrt{\frac{p(x)\, q(x)}{p(\XX_\delta)}}\, F(\rho_x,\omega_x) \geq \frac{\sqrt{1-\e^2} - \delta \sqrt{1-p(\XX_\delta)}}{\sqrt{p(\XX_\delta)}}\, .
\label{quasi_concavity_D_max_eps_proof_eq5}
\ee
Since
\bb
\sum_{x\in \XX_\delta} \sqrt{\frac{p(x)\, q(x)}{p(\XX_\delta)}} \,\leqt{(v)}\, \sqrt{\sumno_{x\in \XX_\delta} \frac{p(x)}{p(\XX_\delta)}}\, \sqrt{\sumno_{x\in \XX_\delta} q(x)} \,\leqt{(vi)}\, 1\, ,
\ee
where (v)~is again by Cauchy--Schwarz while in~(vi) we remembered~\eqref{quasi_concavity_D_max_eps_proof_eq4},  from~\eqref{quasi_concavity_D_max_eps_proof_eq5} it follows that there exists some $\xbar\in \XX_\delta$ with the property that
\bb
F\big(\rho_{\xbar},\omega_{\xbar}\big)\, &\geq\ \frac{\sqrt{1-\e^2} - \delta \sqrt{1-p(\XX_\delta)}}{\sqrt{p(\XX_\delta)}} \\
&\geq\ \inf_{\lambda\in (0,1]} \frac{\sqrt{1-\e^2} - \delta \sqrt{1-\lambda}}{\sqrt{\lambda}} \\
&\eqt{(vii)}\ \sqrt{1-\e^2-\delta^2} \, .
\ee
To prove~(vii), we can introduce the parameter $x\coloneqq \sqrt{\frac{1-\lambda}{\lambda}} \in [0,\infty)$, in terms of which we need to show that
\bb
\inf_{x \in [0,\infty)} \left\{ \sqrt{1-\e^2}\, \sqrt{1+x^2} - \delta x \right\} \eqt{?} \sqrt{1-\e^2-\delta^2}\, .
\ee
This can be swiftly proved by applying the reversed Cauchy--Schwarz inequality~\cite{Aczel1956} in the Minkowski space with scalar product $\lsmatrix 1 & \\ & -1 \rsmatrix$ to the two vectors $\big(\sqrt{1+x^2},\,x\big)$ and $\big(\sqrt{1-\e^2},\, \delta\big)$, observing that saturation occurs in the parallel case. A more traditional approach is as follows: the function on the left-hand side is clearly strictly convex in $x$, as $\sqrt{1+x^2} = \| (1,x) \|_2$, diverges as $x\to +\infty$, and has non-positive derivative at $x=0$; therefore, the minimum exists, is unique, and can be found by setting the derivative to zero. Doing so shows, once again, that the minimum is located at $x = \frac{\delta}{\sqrt{1-\e^2-\delta^2}}$ and has the value given in~(vii).

Going back to the purified distance, we see that
\bb
P\big(\rho_{\xbar},\omega_{\xbar}\big) \leq \sqrt{\e^2+\delta^2}\, .
\label{quasi_concavity_D_max_eps_proof_eq8}
\ee
We can now choose 
\bb
\delta \coloneqq \sqrt{(\e+\mu)^2 - \e^2}\, ,
\label{quasi_concavity_D_max_eps_proof_eq9}
\ee
so that the right-hand side of the above inequality is equal to $\e+\mu$. This lets us rephrase~\eqref{quasi_concavity_D_max_eps_proof_eq8} as
\bb
P\big(\rho_{\xbar},\omega_{\xbar}\big) \leq \e + \mu\, .
\label{quasi_concavity_D_max_eps_proof_eq10}
\ee

Now, we have
\bb
\min_x D_{\max}^{\e+\mu,\ P}\!(\rho_x\|\sigma)\, &\leq\, \rel{D_{\max}^{\e+\mu}}{\rho_{\xbar}}{\sigma} \\
&\leqt{(viii)}\, \rel{D_{\max}}{\omega_{\xbar}}{\sigma} \\
&\leqt{(ix)}\, D_{\max}(\omega\|\sigma) + \log\frac{1}{q(\xbar)} \\
&\leqt{(x)}\, D_{\max}^{\e,P}(\rho\|\sigma) + \log\frac{|\XX|}{\delta^2}\, ,
\label{quasi_concavity_D_max_eps_proof_eq11}
\ee
where (viii)~is a direct consequence of~\eqref{quasi_concavity_D_max_eps_proof_eq10}, (ix)~can be justified by observing that
\bb
\omega_{\xbar} \leq \frac{1}{q(\xbar)} \sum_x q(x)\omega_x = \frac{\omega}{q(\xbar)} \leq \frac{\exp\left[D_{\max}(\omega\|\sigma)\right]}{q(\xbar)}\, \sigma
\ee
(in other words, by using the triangle inequality for the max-relative entropy), and, finally, in~(x) we used~\eqref{quasi_concavity_D_max_eps_proof_eq0} and remembered that $\xbar\in \XX_\delta$, where $\XX_\delta$ is defined in~\eqref{quasi_concavity_D_max_eps_proof_eq2}. Here, $\delta$ is given by~\eqref{quasi_concavity_D_max_eps_proof_eq9}; performing this substitution in~\eqref{quasi_concavity_D_max_eps_proof_eq11} yields precisely the claimed inequality~\eqref{quasi_concavity_D_max_eps} and thereby concludes the proof.
\end{proof}


By using the `weak/strong converse duality' bounds connecting the smoothed max-relative entropy and the hypothesis testing relative entropy, first established in~\cite{Tomamichel2013} and then refined in~\cite{Anshu2019, tight-relations}, we are now able to deduce the following.

\begin{cor} \label{quasi_concavity_D_H_cor}
Let $(p(x),\rho_x)_{x\in \XX}$ be a finite ensemble of states on the same quantum system, and let $\sigma$ be another state on the same system. Then, for all $\e\in (0,1)$ and all $\delta\in (0,\e)$, we have
\bb
D_H^\e \Big(\sumno_{x\in \XX} p(x)\,\rho_x \,\Big\|\, \sigma\Big) &\geq \min_x D_H^{\e-\delta}(\rho_x\|\sigma) - \log \big(|\XX|\, G(\e,\delta)\big)\, , \label{quasi_concavity_D_H}
\ee
where
\bb
G(\e,\delta) &\coloneqq \inf_{\nu\in (0,\delta)} \frac{(1-\e)\, F_2\left(\e-\delta,\, 1-\e+\nu\right)}{\nu(\delta-\nu)^2}\, ,
\ee
with $F_2(p,q) \coloneqq \left(\sqrt{pq} + \sqrt{(1-p)(1-q)}\right)^2$ being the binary fidelity function, is a positive constant that depends only on $\e$ and $\delta$. In particular, 
\bb
D_H^\e \Big(\sumno_{x\in \XX} p(x)\,\rho_x \,\Big\|\, \sigma\Big) &\geq \min_x D_H^{\e-\delta}(\rho_x\|\sigma) - \log |\XX| - \log \frac{27 \e (1-\e) (1-\e+\delta)}{\delta^3}\, . \label{quasi_concavity_D_H_simplified}
\ee
Hence, for $\delta = \kappa \e$ with $\kappa\in (0,1)$, we obtain
\bb
D_H^\e \Big(\sumno_{x\in \XX} p(x)\,\rho_x \,\Big\|\, \sigma\Big) &\geq \min_x D_H^{(1-\kappa)\e}(\rho_x\|\sigma) - \log\frac{|\XX|}{\kappa^3 \e^2}\, , 
\label{quasi_concavity_D_H_weak_converse}
\ee
while, for $\delta = \kappa(1-\e)$ with $\kappa\in (0,1)$,
\bb
D_H^\e \Big(\sumno_{x\in \XX} p(x)\,\rho_x \,\Big\|\, \sigma\Big) &\geq \min_x D_H^{\e-\kappa(1-\e)}(\rho_x\|\sigma) - \log \frac{|\XX| (1+\kappa)}{\kappa^3 (1-\e)}\, . 
\label{quasi_concavity_D_H_strong_converse}
\ee
\end{cor}

\begin{proof}
It follows by combining Theorem~\ref{quasi_concavity_D_max_eps_thm} with the bounds in~\cite[Eq.~(96)]{tight-relations}. Specifically, for some $\nu\in (0,\delta)$ to be fixed later, we write
\bb
D_H^\e \Big(\sumno_{x\in \XX} p(x)\,\rho_x \,\Big\|\, \sigma\Big) &\geqt{(i)} D_{\max}^{\sqrt{1-\e},\,P} \Big(\sumno_{x\in \XX} p(x)\,\rho_x \,\Big\|\, \sigma\Big) + \log\frac{1}{1-\e} \\
&\geqt{(ii)} \min_x D_{\max}^{\sqrt{1 -\e+\nu},\, P}\!(\rho_x\|\sigma) - \log \frac{(1-\e) |\XX|}{\nu} \\
&\geqt{(iii)} \min_x D_H^{\e-\delta}(\rho_x\|\sigma) - \log \frac{F_2\left(\e-\delta,\, 1-\e+\nu\right)}{(\delta-\nu)^2} - \log \frac{(1-\e) |\XX|}{\nu}\, .
\ee
Here, in~(i) we employed the first inequality in~\cite[Eq.~(96)]{tight-relations}, with the substitution $\e\mapsto 1-\e$, (ii)~follows from Theorem~\ref{quasi_concavity_D_max_eps_thm} once one observes that $\nu < \delta <\e$, and, finally, in~(iii) we used the second inequality in~\cite[Eq.~(96)]{tight-relations} with the substitutions $\e\mapsto 1-\e+\delta$ and $\mu \mapsto \delta -\nu$. Optimising over $\nu$ yields the claim~\eqref{quasi_concavity_D_H}. The bound in~\eqref{quasi_concavity_D_H_simplified} is obtained by writing
\bb
F_2\left(\e-\delta,\, 1-\e+\nu\right) &= \left( \sqrt{(\e-\delta)(1-\e+\nu)} + \sqrt{(1-\e+\delta)(\e-\nu)} \right)^2 \\
&\leq \left( \sqrt{\e(1-\e+\delta)} + \sqrt{(1-\e+\delta)\e} \right)^2 \\
&= 4 \e(1-\e+\delta)\, ,
\ee
and by noticing that 
\bb
\inf_{\nu\in (0,\delta)} \frac{4}{\nu(\delta-\nu)^2} = \frac{27}{\delta^3}\, ,
\ee
as an elementary calculation swiftly reveals.
\end{proof}

\section{A first black-box reduction of Stein exponents}\label{sec:reduction}

We now move on to the Stein exponent problem, where 
Corollary~\ref{quasi_concavity_D_H_cor} gives us the following.

\begin{prop} \label{first_Stein_reduction_prop}
Let $\HH$ be a finite-dimensional Hilbert space. Consider two sequences $\RR = ( \RR_n)_n$ and $\SS = (\SS_n)_n$ of sets $\RR_n,\SS_n \subseteq \D\big(\HH^{\otimes n}\big)$ of states on the $n$-copy system. Assume that $\co(\RR_n) = \co(\FF_n)$ can be written as the convex hull of some restricted set $\FF_n \subseteq \RR_n$, and that $\FF_n$, and thus $\RR_n$, contain only permutationally symmetric states, i.e.
\bb
U_\pi^{\vphantom{\dag}} \rho_n U_\pi^\dag = \rho_n \qquad \forall\ \rho_n \in \FF_n\, ,\quad \forall\ \pi\in S_n\, ,
\ee
where $U_\pi$ is the unitary that permutes the tensor factors of $\HH^{\otimes n}$ according to the permutation $\pi$. Then, with the notation in~\eqref{Stein}--\eqref{strong_converse_Stein}, we have
\begin{align}
\stein (\RR\|\SS) &= \lim_{\e\to 0^+} \liminf_{n\to\infty} \frac1n \inf_{\rho_n \in \FF_n} D_H^\e(\rho_n\| \co(\SS_n))\, , \label{first_Stein_reduction} \\
\stein^\dag(\RR \|\SS) &= \lim_{\e\to 1^-} \limsup_{n\to\infty} \frac1n \inf_{\rho\in \FF_n} D_H^\e(\rho_n\| \co(\SS_n))\, . \label{first_Stein_reduction_strong_converse}
\end{align}
In particular, if $\RR_n = \FF_n = \RR_n^{\iid}$ for all $n$, where the notation is defined in~\eqref{composite_iid}, we have
\begin{align}
\rel{\stein}{\RR^\iid}{\SS} &= \lim_{\e\to 0^+} \liminf_{n\to\infty} \frac1n \inf_{\rho \in \RR_1} \rel{D_H^\e}{\rho^{\otimes n}}{\co(\SS_n)}\, , \label{first_Stein_reduction_composite_iid} \\
\rel{\stein^\dag\!}{\RR^\iid}{\SS} &= \lim_{\e\to 1^-} \limsup_{n\to\infty} \frac1n \inf_{\rho \in \RR_1} \rel{D_H^\e}{\rho^{\otimes n}}{\co(\SS_n)}\, . \label{first_Stein_reduction_strong_converse_composite_iid}
\end{align}
\end{prop}

\begin{proof}
To prove that the left-hand sides of~\eqref{first_Stein_reduction} and~\eqref{first_Stein_reduction_strong_converse} cannot be larger than the corresponding right-hand sides, simply write, for all $\e\in (0,1)$,
\bb
\frac1n\, D_H^\e(\co(\RR_n)\| \co(\SS_n)) &\leq \frac1n \inf_{\rho_n \in \FF_n} D_H^\e(\rho_n\| \co(\SS_n))\, ,
\label{first_Stein_reduction_proof_eq1}
\ee
where the inequality follows from the observation that $\FF_n \subseteq \RR_n \subseteq \co(\RR_n)$. Taking the limit infimum as $n\to\infty$ and then the limit $\e\to 0^+$ shows that
\bb
\stein(\RR\|\SS) \leq \lim_{\e\to 0^+} \liminf_{n\to\infty} \frac1n \inf_{\rho_n \in \FF_n} D_H^\e(\rho_n\| \co(\SS_n))\, ;
\label{first_Stein_reduction_proof_eq2}
\ee
an analogous inequality holds for the strong converse Stein exponent, as can be verified by taking, in~\eqref{first_Stein_reduction_proof_eq1}, the limit superior as $n\to \infty$ and then the limit $\e\to 1^-$. It thus remains to prove that the left-hand sides of~\eqref{first_Stein_reduction} and~\eqref{first_Stein_reduction_strong_converse} cannot be smaller than the right-hand sides. 

For some $n\in \N^+$, consider a $\rho_n\in \co(\RR_n) = \co(\FF_n)$, and let $d\coloneqq \dim \HH$ denote the dimension of the underlying Hilbert space. Note that $\FF_n$ is entirely contained in the real vector space $\mathrm{H}_{d,n}^{\mathrm{sym}} \vphantom{\scaleobj{2}{p}}$ of permutationally symmetric Hermitian operators on $\HH^{\otimes n} \simeq \big(\C^d\big)^{\otimes n}$. By Schur--Weyl duality, this has the form~\cite[Eq.~(6.15)]{HAYASHI-GROUP}
\bb
\mathrm{H}_{d,n}^{\mathrm{sym}} = \bigoplus_{\lambda \in \YY^d_n} \pazocal{U}_\lambda \otimes \id_{\pazocal{V}_\lambda}\, ,
\ee
where $\lambda$ is an index ranging on the set $\YY^d_n$ of Young diagrams with size $n$ and depth at most $d$, and $\pazocal{U}_\lambda$ and $\pazocal{V}_\lambda$ are irreps of the special unitary group $\mathrm{SU}(d)$ and of the symmetric group $S_n$, respectively. Since the inequalities~\cite[Eq.~(6.16) and~(6.18)]{HAYASHI-GROUP}
\bb
\dim \pazocal{U}_\lambda \leq (n+1)^{d(d-1)/2} ,\qquad |\YY^d_n| \leq (n+1)^{d-1}
\ee
hold, we obtain that
\bb
N \coloneqq \dim \mathrm{H}_{d,n}^{\mathrm{sym}} \leq (n+1)^{d-1} (n+1)^{d(d-1)/2} = (n+1)^{(d-1)\left(\frac{d}{2}+1\right)} .
\label{first_Stein_reduction_proof_eq5}
\ee
Remembering that $\rho_n\in \co(\FF_n)$, and applying Carath\'eodory's theorem inside the affine space composed by all operators in $\mathrm{H}_{d,n}^{\mathrm{sym}} \vphantom{\scaleobj{1.5}{|}}$ with unit trace, which has dimension $N-1$, we see that we can find a finite collection $(\rho_{n,x})_{x=1,\ldots,N}$ of states $\rho_{n,x}\in \FF_n$ and a corresponding probability distribution $(p_n(x))_{x=1,\ldots,N}$ with the property that
\bb
\rho_n = \sum_{x=1}^N p_n(x)\, \rho_{n,x}\, .
\label{discretisation_perm_symm}
\ee
(This discretisation argument is inspired from~\cite[Lemma~A.3]{brandao_adversarial}, where the special case of a composite i.i.d.\ null hypothesis is tackled; the particular way in which we apply Carath\'eodory's theorem is slightly different here, meaning that we obtain a marginally better estimate on $N$.)

Now, let $\e\in (0,1)$, and pick some $\delta \in (0,\e)$. For any fixed $\sigma_n\in \co(\SS_n)$, we can write
\bb
D_H^\e(\rho_n\|\sigma_n) &= \Rel{D_H^\e}{\sumno_{x=1}^N p_n(x)\, \rho_{n,x}}{\sigma_n} \\
&\geqt{(i)} \min_{x\in \{1,\ldots,N\}} D_H^{\e-\delta}(\rho_{n,x} \| \sigma_n) - \log \frac{N}{G(\e,\delta)} \\
&\geqt{(ii)} \inf_{\omega_n\in \FF_n} D_H^{\e-\delta}(\omega_n \| \co(\SS_n)) - \log \frac{N}{G(\e,\delta)} \\
&\geqt{(iii)} \inf_{\omega_n\in \FF_n} D_H^{\e-\delta}(\omega_n \| \co(\SS_n)) - \log \frac{1}{G(\e,\delta)} - (d-1)\left(\tfrac{d}{2}+1\right) \log(n+1)\, .
\ee
Here, in~(i) we applied Corollary~\ref{quasi_concavity_D_H_cor}, in~(ii) we relaxed the minimisation to arbitrary states $\omega_n\in \FF_n$ and $\sigma_n\in \co(\SS_n)$, and finally in~(iii) we employed the polynomial bound on $N$ reported in~\eqref{first_Stein_reduction_proof_eq5}.

We can now take the infimum on the leftmost side with respect to $\rho_n\in \co(\RR_n)$ and $\sigma_n\in \co(\SS_n)$, which gives us
\bb
&\rel{D_H^\e}{\co(\RR_n)}{\co(\SS_n)} \\
&\qquad \geq \inf_{\omega_n\in \FF_n} D_H^{\e-\delta}(\omega_n \| \co(\SS_n)) - \log \frac{1}{G(\e,\delta)} - (d-1)\left(\tfrac{d}{2}+1\right) \log(n+1)\, .
\label{first_Stein_reduction_proof_eq8}
\ee
Dividing by $n$ and taking the limit infimum as $n\to\infty$ yields immediately
\bb
\liminf_{n\to\infty} \frac1n\, \rel{D_H^\e}{\co(\RR_n)}{\co(\SS_n)} &\geq \liminf_{n\to\infty} \frac1n \inf_{\omega_n\in \FF_n} D_H^{\e-\delta}(\omega_n \| \co(\SS_n)) \\
&\geqt{(iv)} \lim_{\e'\to 0^+} \liminf_{n\to\infty} \frac1n \inf_{\omega_n\in \FF_n} D_H^{\e'}(\omega_n \| \co(\SS_n))\, ,
\ee
where in~(iv) we remembered that the hypothesis testing relative entropy is monotonically non-decreasing in the type-I error parameter for any fixed pair of states. Taking the limit $\e\to 0^+$ proves that
\bb
\stein(\RR\|\SS) &\geq \lim_{\e'\to 0^+} \liminf_{n\to\infty} \frac1n \inf_{\omega_n\in \FF_n} D_H^{\e'}(\omega_n \| \co(\SS_n))\, ,
\ee
which, together with~\eqref{first_Stein_reduction_proof_eq2}, completes  the proof of~\eqref{first_Stein_reduction}.

To establish~\eqref{first_Stein_reduction_strong_converse}, we need to go back to~\eqref{first_Stein_reduction_proof_eq8}, where, upon the usual division by $n$, we now take the limit superior as $n\to\infty$. Doing so leads us to the inequality
\bb
\stein^\dag(\RR\|\SS) &\geq \limsup_{n\to\infty} \frac1n\, \rel{D_H^\e}{\co(\RR_n)}{\co(\SS_n)} \\
&\geq \limsup_{n\to\infty} \frac1n\, \inf_{\omega_n\in \FF_n} D_H^{\e-\delta}(\omega_n \| \co(\SS_n))\, .
\ee
Setting 
\bb
\delta \coloneqq \e(1-\e)\, ,
\ee
we obtain that
\bb
\stein^\dag(\RR\|\SS) &\geq \limsup_{n\to\infty} \frac1n\, \inf_{\omega_n\in \FF_n} D_H^{\e^2}(\omega_n \| \co(\SS_n))\, .
\ee
Defining $\e' \coloneqq \e^2$ and taking the limit $\e\to 1^-$ then shows that
\bb
\stein^\dag(\RR\|\SS) &\geq \lim_{\e'\to 1^-}\limsup_{n\to\infty} \frac1n\, \inf_{\omega_n\in \FF_n} D_H^{\e'}(\omega_n \| \co(\SS_n))\, .
\ee
Since we already established the converse bound (analogous to~\eqref{first_Stein_reduction_proof_eq2}) 
the proof is complete.
\end{proof}

\section{Proof of the main result}\label{sec:mainproof}

The simplified expression in~\eqref{first_Stein_reduction_composite_iid} is promising, as the minimisation over $\co\big(\RR_n^{\iid}\big)$ in the first argument of the hypothesis testing relative entropy has been reduced to an easier-to-handle optimisation over i.i.d.\ states of the form $\rho^{\otimes n}$, where $\rho\in \RR_1$. The beauty of this reduction argument lies in its complete blindness to the alternative hypothesis: Proposition~\ref{first_Stein_reduction_prop} holds with no assumption on the sequence $\SS = (\SS_n)_n$. However, comparing~\eqref{first_Stein_reduction_composite_iid} with~\eqref{goal_Stein_reduction_composite_iid}, we see that we are far from done yet, as the minimum over $\rho$ is still \emph{inside} the two limits over $n$ and $\e$, meaning that the expression in~\eqref{first_Stein_reduction_composite_iid} is difficult to relate to the right-hand side of~\eqref{goal_Stein_reduction_composite_iid}. We will now attempt to swap minimisation and limits. Something similar was achieved in~\cite[Lemma~12]{doubly-comp-quantum} for the better-behaved relative entropy; here, things are significantly complicated by the fact that we have to deal with a one-shot entropy instead. To tackle the problem, we will need some assumptions on the alternative hypothesis, meaning that our approach ceases to be completely blind to the nature of the latter.

We start from the following slightly stronger reformulation of the `asymptotic quantum blurring lemma' from~\cite{GQSL}. The main difference with~\cite[Lemma~11]{GQSL} is that we allow the base state of which we take as asymptotically large number of i.i.d.\ copies to vary with $n$. Fortunately, at this stage it suffices to look at the first steps of the proof of~\cite[Lemma~11]{GQSL} to see that it can easily adapt to cover also the case of interest here.

To set the stage, we remind the reader that, for some finite-dimensional Hilbert space $\HH$, some $n\in \N^+$, some $\delta\in \R^+$, and a state $\rho\in \D(\HH)$, the \deff{$\boldsymbol{\rho}$-dependent blurring map} is defined as~\cite[Eq.~(88)]{GQSL}
\bb
\widebar{B}_{n,\delta}^{\,\rho} &: \LL\big(\HH^{\otimes n} \big) \to \LL\big(\HH^{\otimes n} \big)\, , \\
\widebar{B}_{n,\delta}^{\,\rho} (X) &\coloneqq \Tr_{\floor{n \delta}} \E_{\pi}\, U_\pi^{\vphantom{\dag}}\big(X \otimes \rho^{\otimes \floor{n \delta}} \big) U_\pi^\dag\, ,
\label{blurring_map}
\ee
where $\pi \in S_{n+\floor{n \delta}}$ is a uniformly random permutation of a set of $n+\floor{n \delta}$ elements, and $U_\pi$ denotes the unitary that implements it on the tensor factors of $\HH^{\otimes (n+\floor{n \delta})}$.

\begin{lemma}[{(Slightly stronger form of the quantum blurring lemma~\cite[Lemma~11]{GQSL})}] \label{stronger_q_blurring_lemma}
Let $(\rho_n)_n$ be a sequence of states $\rho_n\in \D(\HH)$ on a fixed finite-dimensional system. For some infinite set $I\subseteq \N$, let $(\Omega_n)_{n\in I}$ be a sequence of permutationally symmetric $n$-copy states $\Omega_n \in \D\big( \HH^{\otimes n}\big)$ such that
\bb
\limsup_{n\in I} \frac12\, \big\|\, \Omega_n - \rho_n^{\otimes n}\big\|_1 \leq \e
\label{non_orthogonal_stronger_q_blurring}
\ee
for some $\e\in (0,1)$. Then there exists an infinite subset $I'\subseteq I$ such that, for all $\Delta\in (0,\frac12]$,
\bb
\lim_{M\to\infty} \limsup_{n\in I'} \Tr \left( \rho_n^{\otimes n} - M \int_0^{\Delta} \frac{\dd \delta}{\Delta}\ \widebar{B}_{n,\delta}^{\,\rho_n}(\Omega_n) \right)_+ = 0\, ,
\label{stronger_q_blurring}
\ee
where the channel $\widebar{B}_{n,\delta}^{\,\rho_n}$ is defined as in~\eqref{blurring_map}.
\end{lemma}

\begin{proof}
We can proceed exactly as in~\cite[Section~V.I, proof of Lemma~11]{GQSL} to reduce the problem to the case where $\rho_n$ is pure for all $n$. Indeed, suppose that the claim has been shown in this case, and let us show that it holds in the mixed state case as well. By~\cite[Lemma~29]{GQSL}, for an arbitrary $n\in \N^+$ we can consider two purifications $\ket{\psi_n}_{AE}$ of $\rho_n = \rho_n^A$ and $\ket{\Phi_n}_{A^nE^n}$ of $\Omega_n = \Omega_n^{A^n}$, the latter supported on the fully symmetric space of $A^nE^n = (AE)^n$, so that 
\bb
\Tr_{E^n} \psi_n^{\otimes n} &= \Tr_{E^n} \ketbra{\psi_n}^{\otimes n}_{AE} = \big(\rho_n^{\otimes n}\big)_{A^n}\, ,\\
\Tr_{E^n} \Phi_n &= \Tr_{E^n} \ketbra{\Phi_n}_{A^nE^n} = \Omega_n^{A^n}\! ,
\label{partial_traces_symm_purifications}
\ee
and furthermore
\bb
\big|\braket{\psi_n^{\otimes n}|\Phi_n}\big| \geq 1 - \frac12 \left\|\rho_n^{\otimes n} - \Omega_n\right\|_1 .
\ee
From~\eqref{non_orthogonal_stronger_q_blurring}, we thus get that
\bb
\liminf_{n\in I} \big|\braket{\psi_n^{\otimes n}|\Phi_n}\big| \geq 1 - \limsup_{n\in I} \frac12 \left\|\rho_n^{\otimes n} - \Omega_n\right\|_1 \geq 1-\e\, ;
\ee
because of the well-known expression for the trace distance between pure states, this entails that
\bb
\limsup_{n\in I} \frac12 \left\| \psi_n^{\otimes n} - \Phi_n \right\|_1 = \limsup_{n\in I} \sqrt{1- \big|\braket{\psi_n^{\otimes n} | \Phi_n}\big|^2} \leq \sqrt{1-(1-\e)^2} \eqqcolon \e' \in (0,1)\, .
\label{epsilon_prime}
\ee

Now, exactly as in~\cite[Eq.~(214)]{GQSL}, by tracing away all the $E$ systems it is easy to verify that
\bb
\Tr_{E^n} \widebar{B}^{\,\psi_n}_{n,\delta}(\Phi_n) = \widebar{B}^{\,\rho_n}_{n,\delta} (\Omega_n)\, .
\label{partial_trace_output_blurring}
\ee
Therefore, for any infinite subset $I'\subseteq I$, to be fixed shortly, we have
\bb
&\lim_{M\to\infty} \limsup_{n \in I'} \Tr \left(\rho_n^{\otimes n} - M \int_0^\Delta \frac{\dd\delta}{\Delta}\, \widebar{B}_{n,\delta}^{\,\rho_n}(\Omega_n) \right)_+ \\
&\quad = \lim_{M\to\infty} \limsup_{n \in I'} \Tr\left( \Tr_{E^n} \left[ \psi_n^{\otimes n} - M \int_0^\Delta \frac{\dd\delta}{\Delta}\, \widebar{B}_{n,\delta}^{\,\psi_n}(\Phi_n) \right] \right)_{+} \\
&\quad \leq \lim_{M\to\infty} \limsup_{n \in I'} \Tr \left( \psi_n^{\otimes n} - M \int_0^\Delta \frac{\dd\delta}{\Delta}\, \widebar{B}_{n,\delta}^{\,\psi_n}(\Phi_n) \right)_{+} \\
&\quad =\ 0\, .
\ee
Here, the first equality follows from~\eqref{partial_trace_output_blurring}, the inequality holds by data processing, and the last equality requires the assumption that we have established the result for pure states, for some particular $I'$.

By the above calculation, we can therefore consider, without loss of generality, the case where $\rho_n = \ketbra{\psi_n} = \psi_n\in \D(\HH)$ is pure for all $n$. (We can also assume $\Omega_n$ to be pure, although we will not need this particular fact.) Let us construct a sequence $(V_n)_n$ of unitaries $V_n$ on $\HH$ with the property that $V_n^{\vphantom{\dag}} \psi_n V_n^\dag = \ketbra{0}$ is equal to a fixed pure state $\ketbra{0} \in \D(\HH)$. Defining
\bb
\Omega'_n \coloneqq V_n^{\otimes n} \Omega_n \big(V_n^{\dag}\big)^{\otimes n} ,
\ee
we have, by unitary invariance,
\bb
\limsup_{n\in I} \frac12 \left\|\, \Omega'_n - \ketbra{0}^{\otimes n} \right\|_1 = \limsup_{n\in I} \frac12 \left\|\, \Omega_n - \psi_n^{\otimes n} \right\|_1 \leq \e'\, .
\ee

Since $\Omega'_n$ is permutationally symmetric for all $n$, the quantum blurring lemma~\cite[Lemma~11]{GQSL} ensures the existence of an infinite set $I'\subseteq I$ such that, for all $\Delta\in \big(0,\tfrac12\big]$,
\bb
\lim_{M\to\infty} \limsup_{n\in I'} \Tr \left(\ketbra{0}^{\otimes n} - M \int_0^\Delta \frac{\dd \delta}{\Delta}\, \widebar{B}_{n,\delta}^{\,\ket{0}\!\bra{0}}\left(\Omega'_n\right) \right)_+ = 0\, .
\ee
Again by unitary invariance, we thus obtain that
\bb
0 &= \lim_{M\to\infty} \limsup_{n\in I'} \Tr \left(\big(V_n^\dag\big)^{\otimes n}\ketbra{0}^{\otimes n} V_n^{\otimes n} - M \int_0^\Delta \frac{\dd \delta}{\Delta}\, \big(V_n^\dag\big)^{\otimes n} \widebar{B}_{n,\delta}^{\,\ket{0}\!\bra{0}}\left(\Omega'_n\right) V_n^{\otimes n} \right)_+ \\
&= \lim_{M\to\infty} \limsup_{n\in I'} \Tr \left(\psi_n^{\otimes n} - M \int_0^\Delta \frac{\dd \delta}{\Delta}\, \widebar{B}_{n,\delta}^{\,\psi_n}\left(\Omega_n\right) \right)_+ ,
\ee
where, in the second line, we observed that
\bb
\big(V_n^\dag\big)^{\otimes n} \widebar{B}_{n,\delta}^{\,\ket{0}\!\bra{0}}\left(\Omega'_n\right) V_n^{\otimes n} &= \big(V_n^\dag\big)^{\otimes n} \Tr_{\floor{n \delta}}\left[ \E_{\pi}\, U_\pi^{\vphantom{\dag}}\big(\Omega'_n \otimes \ketbra{0}^{\otimes \floor{n \delta}} \big) U_\pi^\dag \right] V_n^{\otimes n} \\
&= \Tr_{\floor{n \delta}}\left[ \big(V_n^\dag\big)^{\otimes (n+\floor{n \delta})}\, \E_{\pi}\, U_\pi^{\vphantom{\dag}}\big(\Omega'_n \otimes \ketbra{0}^{\otimes \floor{n \delta}} \big) U_\pi^\dag V_n^{\otimes (n+\floor{n \delta})} \right] \\
&= \Tr_{\floor{n \delta}}\left[ \E_{\pi}\, U_\pi^{\vphantom{\dag}}\Big(\big(V_n^\dag\big)^{\otimes (n+\floor{n \delta})}\big( \Omega'_n \otimes \ketbra{0}^{\otimes \floor{n \delta}} \big) V_n^{\otimes (n+\floor{n \delta})} \Big) U_\pi^\dag \right] \\
&= \Tr_{\floor{n \delta}}\left[ \E_{\pi}\, U_\pi^{\vphantom{\dag}}\Big(\Big( \big(V_n^\dag\big)^{\otimes n} \Omega'_n V_n^{\otimes n}\Big) \otimes \big(V_n^\dag \ketbra{0} V_n^{\vphantom{\dag}}\big)^{\otimes \floor{n \delta}} \Big) U_\pi^\dag \right] \\
&= \Tr_{\floor{n \delta}}\left[ \E_{\pi}\, U_\pi^{\vphantom{\dag}}\Big(\Omega_n \otimes \psi_n^{\otimes \floor{n \delta}} \Big) U_\pi^\dag \right] \\
&= \widebar{B}_{n,\delta}^{\,\psi_n}(\Omega_n)\, .
\ee
This completes the proof in the case where $\rho_n$ is pure for all $n$, and, by the above argument, this suffices to establish the claim in general.
\end{proof}

We are finally ready to prove our first main result.

\begin{boxed}
\begin{manualthm}{\ref{black_box_Stein_iid_thm}}
Let $\RR_1\subseteq \D(\HH)$ be any set of states on the finite-dimensional Hilbert space $\HH$. Assume that the sequence $\SS = (\SS_n)_n$ of alternative hypotheses $\SS_n\subseteq \D(\HH^{\otimes n})$ satisfies all five Brand\~{a}o--Plenio axioms. Then
\begin{equation}
    \lim_{n\to\infty} \frac1n\, \rel{D_H^\e}{\co\big(\RR_n^{\iid}\big)}{\SS_n} = \inf_{\rho\in \RR_1} D^\infty(\rho \|\SS)\qquad \forall\ \e\in (0,1)\, ,
\tag{\ref{black_box_Stein_iid_DH}}
\end{equation}
where $\RR_n^{\mathrm{iid}}$ is defined by~\eqref{composite_iid}. Equivalently,
\begin{equation}
\rel{\stein}{\RR^{\mathrm{iid}}}{\SS} = \rel{\stein^\dag}{\RR^{\mathrm{iid}}}{\SS} = \inf_{\rho\in \RR_1} D^\infty(\rho\|\SS) = \inf_{\rho\in \RR_1} \stein(\rho\|\SS)\, .
\tag{\ref{black_box_Stein_iid}}
\end{equation}
\end{manualthm}
\end{boxed}

\begin{proof}
The fact that~\eqref{black_box_Stein_iid_DH} and~\eqref{black_box_Stein_iid} are equivalent is obvious, so, without loss of generality, we focus on~\eqref{black_box_Stein_iid}. Using~\eqref{first_Stein_reduction_composite_iid} and~\eqref{first_Stein_reduction_strong_converse_composite_iid}, and remembering that $\SS_n$ is convex for all $n$ because of the Brand\~{a}o--Plenio axioms, we see that it suffices to show that
\bb
\liminf_{n\to\infty} \frac1n\, \inf_{\rho\in \RR_1} \rel{D_H^\e}{\rho^{\otimes n}}{\SS_n} \geqt{?} \inf_{\rho\in \RR_1} D^\infty(\rho\|\SS) \qquad \forall\ \e\in (0,1)\, .
\label{black_box_Stein_iid_proof_eq1}
\ee
Indeed, with this we would obtain that
\bb
\inf_{\rho\in \SS_1} D^\infty(\rho\|\SS) &\leq \lim_{\e\to 0^+} \liminf_{n\to\infty} \frac1n\, \inf_{\rho\in \RR_1} \rel{D_H^\e}{\rho^{\otimes n}}{\SS_n} \\
&=\, \rel{\stein}{\RR^\iid}{\SS} \\
&\leq\hspace{1pt} \rel{\stein^\dag\!}{\RR^\iid}{\SS} \\
&= \lim_{\e\to 1^-} \limsup_{n\to\infty} \frac1n\, \inf_{\rho\in \RR_1} \rel{D_H^\e}{\rho^{\otimes n}}{\SS_n} \\
&\leq \inf_{\rho \in \RR_1} \lim_{\e\to 1^-} \limsup_{n\to\infty} \frac1n\, \rel{D_H^\e}{\rho^{\otimes n}}{\SS_n} \\
&= \inf_{\rho\in \SS_1} D^\infty(\rho\|\SS) \\
&= \inf_{\rho\in \SS_1} \stein(\rho\|\SS)\, ,
\label{black_box_Stein_iid_proof_eq2}
\ee
where on the second and fourth lines we used~\eqref{first_Stein_reduction_composite_iid} and~\eqref{first_Stein_reduction_strong_converse_composite_iid} from Proposition~\ref{first_Stein_reduction_prop}, and the last two equalities are from (the strong converse form of) the generalised quantum Stein's lemma~\cite{Hayashi-Stein, GQSL}. Provided that we can establish~\eqref{black_box_Stein_iid_proof_eq1}, all inequalities in~\eqref{black_box_Stein_iid_proof_eq2} would actually be equalities, thereby completing the proof.

As usual, to establish~\eqref{black_box_Stein_iid_proof_eq1} we will leverage the weak/strong converse duality~\cite[Theorem~12]{tight-relations}, which gives us the equivalent statement
\bb
\liminf_{n\to\infty} \frac1n \inf_{\rho\in \RR_1} \rel{D_{\max}^\e}{\rho^{\otimes n}}{\SS_n} \geqt{?} \inf_{\rho\in \RR_1} D^\infty(\rho\|\SS) \qquad \forall\ \e\in (0,1)\, .
\ee
We proceed by contradiction. Assume that
\bb
\liminf_{n\to\infty} \frac1n \inf_{\rho\in \RR_1} \rel{D_{\max}^\e}{\rho^{\otimes n}}{\SS_n} < \lambda < \inf_{\rho\in \RR_1} D^\infty(\rho\|\SS)
\label{black_box_Stein_iid_proof_eq4}
\ee
for some $\e\in (0,1)$. This means that there exists an infinite set $J\subseteq \N$ such that
\bb
\inf_{\rho\in \RR_1} \rel{D_{\max}^\e}{\rho^{\otimes n}}{\SS_n} < n\lambda \qquad \forall\ n\in J\, .
\ee
In turn, this entails the existence of a sequence of states $(\rho_n)_{n\in J}$ such that $\rho_n\in \RR_1$ for all $n\in J$, and moreover 
\bb
\rel{D_{\max}^\e}{\rho_n^{\otimes n}}{\SS_n} < n\lambda \qquad \forall\ n\in J\, .
\ee
Using the definition of smooth max-relative entropy, we can now obtain states $\Omega_n\in \D\big(\HH^{\otimes n}\big)$ satisfying
\bb
\frac12 \left\|\,\Omega_n - \rho_n^{\otimes n}\right\|_1 \leq \e\, ,\qquad D_{\max}(\Omega_n\|\SS_n) \leq n\lambda\qquad \forall\ n\in J\, .
\label{black_box_Stein_iid_proof_eq7}
\ee
Due to the fact that $\SS_n$ is closed under permutations of the systems by the Brand\~ao--Plenio axioms, we can assume that $\Omega_n$ is permutationally symmetric for all $n\in J$.

Using the compactness of the set of states, we can also extract from $(\rho_n)_{n\in J}$ a sub-sequence $(\rho_n)_{n\in I}$, where $I\subseteq J \subseteq \N$ is also infinite, such that $\rho_n \tends{}{n\in I\,} \widebar{\rho} \in \widebar{\RR}_1$, with $\widebar{\RR}_1$ being the topological closure of $\RR_1$. Since 
\bb
\limsup_{n\in I} \frac12\, \big\|\, \Omega_n - \rho_n^{\otimes n}\big\|_1 \leq \limsup_{n\in J} \frac12\, \big\|\, \Omega_n - \rho_n^{\otimes n}\big\|_1 \leq \e\, ,
\ee
we can apply Lemma~\ref{stronger_q_blurring_lemma}, which establishes the existence of an infinite set $I'\subseteq I \subseteq J$ such that~\eqref{stronger_q_blurring} holds for all $\Delta\in \big(0,\tfrac12\big]$. 

Now, fix some $\Delta\in \big(0,\tfrac12\big]$ and, for an arbitrary $\delta \in [0,\Delta]$, consider the modified blurring map
\bb
B_{n,\delta} \coloneqq \widebar{B}_{n,\delta}^{\,\tau}\, ,
\label{modified_blurring}
\ee
where the right-hand side is defined according to~\eqref{blurring_map}, and $\tau\in \D(\HH)$ is a full-rank state satisfying
\bb
\tau \geq c\id > 0\, ;
\ee
the existence of such a state, naturally, is guaranteed by the Brand\~ao--Plenio axioms. The same axioms also ensure that $B_{n,\delta}$ is a resource non-generating map, in the sense that
\bb
B_{n,\delta}(\SS_n) \subseteq \SS_n\, .
\label{black_box_Stein_iid_proof_eq11}
\ee
Also, since the state $\tau$ appears exactly $\floor{n \delta} \leq n \Delta$ times in the action of $B_{n,\delta}$, employing the operator lower bound $\tau \geq c \rho_n$ on each of those states shows that
\bb
B_{n,\delta} \succeq c^{n \Delta} \widebar{B}^{\,\rho_n}_{n,\delta}
\ee
with respect to the completely positive ordering. Plugging this into~\eqref{stronger_q_blurring} and using the monotonicity of the trace of the positive part with respect to the positive semi-definite ordering yields 
\bb
\lim_{M\to\infty} \limsup_{n\in I'} \Tr \left( \rho_n^{\otimes n} - \frac{M}{c^{n \Delta}}\, \widetilde{\Omega}_n \right)_+ = 0\, ,\qquad \widetilde{\Omega}_n \coloneqq \int_0^{\Delta} \frac{\dd \delta}{\Delta}\ B_{n,\delta}(\Omega_n)\, .
\ee
Due to~\eqref{black_box_Stein_iid_proof_eq11}, the max-relative entropy distance from $\SS_n$ is monotonically non-increasing under $B_{n,\delta}$; using also~\eqref{black_box_Stein_iid_proof_eq7}, we see that
\bb
\Rel{D_{\max}}{\widetilde{\Omega}_n}{\SS_n} \leq \rel{D_{\max}}{\Omega_n}{\SS_n} \leq n\lambda \qquad \forall\ n\in I'\, .
\label{black_box_Stein_iid_proof_eq14}
\ee
For all fixed $M$ and all sufficiently large $n\in I'$, we can also write $M\leq n$, which leads us to
\bb
\limsup_{n\in I'} \Tr \left( \rho_n^{\otimes n} - \frac{n}{c^{n \Delta}}\, \widetilde{\Omega}_n \right)_+ = 0\, .
\label{black_box_Stein_iid_proof_eq15}
\ee 

We can now proceed for a while essentially mimicking the proof of the generalised quantum Stein's lemma using the asymptotic quantum blurring lemma, reported in~\cite[Section~V.C]{GQSL}. Pick some $\eta>0$. Eq.~\eqref{black_box_Stein_iid_proof_eq15} implies that
\bb
\rel{\widetilde{D}_{\max}^{\eta^2}}{\rho_n^{\otimes n}}{\widetilde{\Omega}_n} \leq n \Delta \log \tfrac1c + \log n
\label{black_box_Stein_iid_proof_eq16}
\ee
for all sufficiently large $n\in I'$, where the Datta--Leditzkly smooth max-relative entropy is defined by
\bb
\widetilde{D}_{\max}^{\zeta}(\rho\|\sigma) \coloneqq \inf\left\{ \lambda:\ \Tr\big(\rho - \exp[\lambda]\, \sigma\big)_+ \leq \zeta \right\} .
\ee
The Datta--Renner lemma, originally due to~\cite[Lemma~3]{Datta2009} (see also~\cite[Lemma~C.3]{Brandao2010}) and recently sharpened in~\cite{tight-relations}, gives us
\bb
D_{\max}^{\eta}(\rho\|\sigma) \leq \widetilde{D}_{\max}^{\eta^2}(\rho\|\sigma) + \log\tfrac{1}{1-\eta^2}\, ,
\label{black_box_Stein_iid_proof_eq18}
\ee
as one can verify immediately by substituting $\e\mapsto \eta^2$ into the first line of~\cite[Corollary~8, Eq.~(56)]{tight-relations}. By chaining~\eqref{black_box_Stein_iid_proof_eq16} and~\eqref{black_box_Stein_iid_proof_eq18} (with the substitutions $\rho \mapsto \rho_n^{\otimes n}$ and $\sigma \mapsto \widetilde{\Omega}_n$), we derive
\bb
\rel{D_{\max}^{\eta}}{\rho_n^{\otimes n}}{\widetilde{\Omega}_n} \leq n \Delta \log \tfrac1c + \log n + \log \tfrac{1}{1-\eta^2}\, ,
\label{black_box_Stein_iid_proof_eq19}
\ee
which is valid for all sufficiently large $n\in I'$. Using the chain rule for the max-relative entropy, we arrive at the inequality
\bb
\rel{D_{\max}^\eta}{\rho_n^{\otimes n}}{\SS_n} &\leq \rel{D_{\max}^{\eta}}{\rho_n^{\otimes n}}{\widetilde{\Omega}_n} + \Rel{D_{\max}}{\widetilde{\Omega}_n}{\SS_n} \\
&\leq n \Delta \log \tfrac1c + \log n + \log \tfrac{1}{1-\eta^2} + n\lambda\, ,
\label{black_box_Stein_iid_proof_eq20}
\ee
again valid for all sufficiently large $n\in I'$, and where, on the second line, we combined~\eqref{black_box_Stein_iid_proof_eq19} (for the first term) and~\eqref{black_box_Stein_iid_proof_eq14} (for the second). Lower bounding the max-relative entropy with the Umegaki relative entropy and using the asymptotic continuity of the relative entropy distance $D(\cdot\|\SS_n)$, due to~\cite{Alicki-Fannes, tightuniform} and reported in~\cite[Lemma~13]{GQSL}, yields, precisely as in~\cite[Eq.~(14)]{GQSL},
\bb
n \Delta \log \tfrac1c + \log n + \log \tfrac{1}{1-\eta^2} + n\lambda \geq \rel{D}{\rho_n^{\otimes n}}{\SS_n} - n\eta \log\tfrac1c - g(\eta)\, ,
\ee
where $g(x) \coloneqq (x+1)\log(x+1) - x \log x$. Dividing by $n$ and re-arranging gives
\bb
\lambda + (\eta + \Delta) \log\tfrac1c + \frac1n \left( \log n + g(\eta) + \log\tfrac{1}{1-\eta^2}\right) \geq \frac1n\, \rel{D}{\rho_n^{\otimes n}}{\SS_n}\, .
\label{black_box_Stein_iid_proof_eq21}
\ee

The last tricky step comes now: in the original proof in~\cite{GQSL}, there is nothing else to do, as the right-hand side features a fixed state $\rho$ and thus converges to the regularised relative entropy distance $D^\infty(\rho\|\SS)$. Here, instead, we have to take care of the fact we have $\rho_n^{\otimes n}$ instead of $\widebar{\rho}^{\,\otimes n}$; plain asymptotic continuity cannot help us now, because, although $\rho_n \tends{}{n\in I'\,} \widebar{\rho}$, in general $\frac12 \left\|\rho_n^{\otimes n} - \widebar{\rho}^{\,\otimes n} \right\|_1$ does not tend to zero along $n\in I'$. To remedy this, we need to use a stronger continuity property of the relative entropy distance from $\SS$, namely, the fact that
\bb
\frac1n \left| \rel{D}{\rho^{\otimes n}}{\SS_n} - \rel{D}{\omega^{\otimes n}}{\SS_n} \right| \leq \xi \log \tfrac1c + g(\xi)
\label{black_box_Stein_iid_proof_eq22}
\ee
provided that $\frac12 \|\rho - \omega\|_1\leq \xi$. The history of~\eqref{black_box_Stein_iid_proof_eq22} is somewhat long. The basic technique is due to Alicki and Fannes~\cite{Alicki-Fannes}, and Christandl figured out a way to apply it to prove the continuity of the regularised relative entropy of entanglement~\cite[Proposition~3.23]{MatthiasPhD}. (A slight generalisation of the proof by Christandl to the case of arbitrary resources obeying the Brand\~ao--Plenio axioms would already suffice to recover an inequality that, while weaker than~\eqref{black_box_Stein_iid_proof_eq22}, would be asymptotically strong enough for us.) More than a decade later, a refined version of the basic Alicki--Fannes technique was put forth by Winter, who used it to obtain a cleaner proof of the continuity of the regularised relative entropy of entanglement, with tighter bounds~\cite[Corollary~8]{tightuniform}. It is not difficult to take his proof and generalise it to arbitrary resources under the Brand\~ao--Plenio axioms, as was done in~\cite[Proof of Lemma~12]{doubly-comp-quantum}. Doing so gives precisely~\eqref{black_box_Stein_iid_proof_eq22} (cf.,\ for instance,~\cite[Eq.~(107)]{doubly-comp-quantum}).

We can now use~\eqref{black_box_Stein_iid_proof_eq22}, with the substitutions $\rho\mapsto \widebar{\rho}$, $\omega \mapsto \rho_n$, and 
\bb
\xi \mapsto \xi_n \coloneqq \frac12 \left\|\widebar{\rho} - \rho_n\right\|_1 \tends{}{n\in I'\,} 0\, ,
\label{black_box_Stein_iid_proof_eq23}
\ee
to continue from~\eqref{black_box_Stein_iid_proof_eq21} and write
\bb
\lambda + (\eta + \Delta) \log\tfrac1c + \frac1n \left( \log n + g(\eta) + \log\tfrac{1}{1-\eta^2}\right) &\geq \frac1n\, \rel{D}{\rho_n^{\otimes n}}{\SS_n} \\
&\geq \frac1n\, \rel{D}{\widebar{\rho}^{\,\otimes n}}{\SS_n} - \xi_n \log\tfrac1c - g(\xi_n)\, .
\ee
We are finally ready to take the limit along $n\in I'$, which yields, due to~\eqref{black_box_Stein_iid_proof_eq23},
\bb
\lambda + (\eta + \Delta) \log\tfrac1c \geq D^\infty(\widebar{\rho}\|\SS) \geq \inf_{\rho \in \widebar{\RR}_1} D^\infty(\rho\|\SS) = \inf_{\rho \in \RR_1} D^\infty(\rho\|\SS)\, ,
\label{black_box_Stein_iid_proof_eq25}
\ee
where, for the last equality, we used the continuity of the function $D^\infty(\cdot\|\SS)$, which follows directly from~\eqref{black_box_Stein_iid_proof_eq22}, to replace the topological closure $\widebar{\RR}_1$ with $\RR_1$ itself. Since $\eta,\Delta>0$ can be made arbitrarily small, we can now take them to zero. Doing so gives
\bb
\lambda \geq \inf_{\rho \in \RR_1} D^\infty(\rho\|\SS)\, ,
\ee
in blatant contradiction with~\eqref{black_box_Stein_iid_proof_eq4}.
\end{proof}

A notable consequence of the above result is the following result, which underlies our analysis of universal quantum resource transformations in the main text.

\begin{cor} \label{continuity_Stein_iid_cor}
Let $\HH$ be a finite-dimensional Hilbert space, and let $\SS = (\SS_n)_n$ be a sequence of sets of states $\SS_n\subseteq \D(\HH^{\otimes n})$ that satisfies all five Brand\~{a}o--Plenio axioms. For some $\rho\in \D(\HH)$, denote with $B_\e(\rho)$ the set of states at trace norm distance at most $\e$ from $\rho$. Then
\bb
\lim_{\e\to 0^+} \rel{\stein}{B_\e(\rho)^\iid}{\SS} = \stein(\rho\|\SS) = D^\infty(\rho\|\SS)\, .
\ee
\end{cor}

\end{document}